\documentclass{article}

\usepackage{geometry}
\usepackage{latexsym,graphicx,amsmath,amsfonts}
\usepackage{algorithmic}
\usepackage{subfig}
\usepackage{multirow}
\usepackage{xr}
\usepackage{natbib}

\usepackage{times}
\usepackage{bm}

\usepackage[figuresright]{rotating}
\usepackage{epstopdf}

\usepackage{amsthm}

\ifpdf
  \DeclareGraphicsExtensions{.eps,.pdf,.png,.jpg}
\else
  \DeclareGraphicsExtensions{.eps}
\fi

\newcommand{\TheTitle}{Experimental Design for Partially Observed Markov Decision Processes}

\title{{\TheTitle}\thanks{This work was funded by NSF grants DMS-1053252 and DEB-1353039}}

\author{
  Leifur Thorbergsson\thanks{Department of Statistical Science, Cornell University, Ithaca, NY
   \underline{lt274@cornell.edu}.}
  \and
  Giles Hooker\thanks{Department of Biological Statistics and Computational Biology, Cornell University, Ithaca NY
    \underline{gjh27@cornell.edu}}.
}

\newcommand{\argmax}{\ensuremath{\operatornamewithlimits{argmax}}}

\newcommand{\fbold}    {\mbox{\bf f}}
\newcommand{\gbold}    {\mbox{\bf g}}

\newcommand{\xbold}    {\mbox{\bf x}}
\newcommand{\ybold}    {\mbox{\bf y}}

\newcommand{\epsilonbold} {\mbox{\boldmath${\epsilon}$}}

\newtheorem{theorem}{Theorem}
\newtheorem{mydef}{Definition}
\newtheorem{myassump}{Assumption}
\newtheorem{mylemma}{Lemma}
\newtheorem{mycorollary}{Corollary}

\begin{document}

\sloppy


%


\maketitle

\begin{abstract}
This paper deals with the question of how to most effectively conduct experiments in Partially Observed Markov Decision Processes so as to provide data that is most informative about a parameter of interest. Methods from Markov decision processes, especially dynamic programming, are introduced and then used in an algorithm to maximize a relevant Fisher Information. The algorithm is then applied to two POMDP examples. The methods developed can also be applied to stochastic dynamical systems, by suitable discretization, and we consequently show what control policies look like in the Morris-Lecar Neuron model, and simulation results are presented. We discuss how parameter dependence within these methods can be dealt with by the use of priors, and develop tools to update control policies online. This is demonstrated in another stochastic dynamical system describing growth dynamics of DNA template in a PCR model.
\end{abstract}

%
%

\section{Introduction} \label{sec:intro}

Hidden Markov Models have proven their usefulness across a wide variety of applications. In many of these applications, the user or the experimenter will have some way of influencing the transitions of the underlying Markov Chain, as in Markov Decision Processes, and such a process is called a Partially Observed Markov Decision Process (POMDP), see \cite{POMDP}. In this paper we assume that the transition probability matrix is governed by an unknown parameter $\theta$ and we wish to understand how the process can be influenced to obtain that will be most informative about $\theta$. We think of this as experimental design for Partially Observed Markov Decision Processes.

Formally, we consider a POMDP $(x_t, y_t, u_t)_{t=0,\ldots, T}$. In this setting $x_t$ is an unobserved Markov Chain, where the transition probabilities depend parametrically on the control $u_t$ that is chosen at time $t$ and an unknown parameter $\theta$. 
$x_t$ is not directly observable. Instead, the process $y_t$ is observed. We assume that $y_t$ depends on which state $x_t$ is in but that given $x_t$ the $y_t$'s are independent over time.

Our goal is to find ways to use the controls $u_t$ to improve parameter estimates of $\theta$. Our general strategy will be to use the controls to try to minimize the sample variance of the maximum likelihood estimates of $\theta$. This will be achieved by maximizing an approximation to the Fisher Information for $\theta$. The controls are calculated using dynamic programming, a popular maximization algorithm from Markov Decision Processes which outputs  an adaptive control policy, i.e. the control chosen at time $t$ is based on observations up to time $t$.

The first attempt at using dynamic controls to maximize a Fisher Information is given in \cite{FOFI}. This paper considered a Markov decision process  in which $x_t$ is directly observed and obtained controls to maximize the Fisher Information for this case. When $x_t$ is not directly observed, an estimate $\hat{x}_t$ was obtained from the observations $y_t$ by applying a particle filter and then employing the control policy that was calculated as though $x_t$ was directly observed. This objective was referred to as the ``Full Observation Fisher Information'' (FOFI) in recognition that it did not correspond to the Fisher Information for the POMDP, although it was argued that maximizing FOFI would still be a useful strategy when $y_t$ was informative about $x_t$.

This paper extends this work in directly using the POMDP structure where the Fisher Information can be calculated recursively. However, the policy for choosing $u_t$ depends on the entire history of the observations $y_t$ and cannot be tractably calculated or stored. Instead, we base the policy on the past $m$ values of $(y_t, u_{t-1})$ and show that this approximates the Fisher Information. We label the resulting approximation the ``Partial Observation Fisher Information'' (POFI) to both distinguish it from FOFI and to note that each control is obtained using only a small number of recent observations.   The control policies based on POFI and FOFI are compared and we illustrate a setting in which the control policies are quite different.


The methods developed here have application value beyond Partially Observed Markov Decision Processes. The methods in \cite{FOFI} were explicitly developed for a diffusion process of the form
\[
d \xbold = \fbold(\xbold,\theta,u(t)) dt + \Sigma^{1/2} d\mathbf W
\]
where $\theta$ is the parameter of interest, to be estimated, $u(t)$ is a control that can be chosen by the user, $\xbold$ is the vector of state variables, $\fbold$ is a vector valued function, $\mathbf W$ a Wiener process, and additionally $\xbold(t)$ is only observed partially or noisily with a continuous-valued observation $y_t$. By discretizing time, state, and observation spaces, the process can be approximated by a POMDP, allowing us to use the methods developed in this paper to devise a control policy that maximizes information about the parameter $\theta$.

The methods we use to calculate controls for maximizing Fisher Information will depend on the unknown parameter $\theta$. We illustrate how this problem can partially be overcome by assuming a prior for $\theta$ to calculate a control policy before running the experiment. Additionally we describe how, using data acquired as the experiment progresses, a posterior for $\theta$ can be used to calculate a more precise control policy. That is, parameter information from observations acquired at a time $t$ can be used to improve the policy used in what is left of the experiment. These methods will be based on the Value Iteration Algorithm (VIA), which is closely related to dynamic programming.

This paper contains three separate developments: (i) the development of a recursive formulation of the Fisher Information for POMDPs and resulting POFI approximation (Sections \ref{sec:framework} and \ref{sec:discreteexamples}), (ii) discretization methods to approximate continuous-valued processes by POMDPs (Section \ref{sec:discretization}) and (iii) the use of a prior to average the POFI or FOFI over possible parameter values when calculating optimal designs and the development of the Value Iteration Algorithm to allow the prior to be updated as the experiment progresses (Section \ref{sec:paramdependence}). Note that in developments (i) and (ii) we have assumed a single value of the unknown parameter $\theta$ to clarify the exposition.

In order to illustrate our methods we present four examples to illustrate each of these developments. Section \ref{sec:discreteexamples} presents two POMDP's. In the first, we hypothesize about the kind of systems in which we will observe large improvement in parameter estimation by using the POFI control policy over the FOFI policy. Following a discussion, we construct a mock Partially Observed Markov Decision Process, in which this improvement is shown using a simulation study.
To illustrate the real-world applicability of design in discrete POMDP's we consider a realistic POMDP from experimental economics. The model will consist of a simple adversarial game similar to the ``rock - paper - scissor" game where one player tries to play in such a way that maximizes information about the other players' strategy.

In Section \ref{subsec:ml} we illustrate the discretization of a continuous-valued system using a stochastic version of the Morris-Lecar Neuron model, a dynamical system which models voltage in a single neural cell. This model is two dimensional, but only one dimension is observed. The model has multiple parameters and we sequentially investigate how the POFI and FOFI control policies perform in estimating each of them in turn.

In Section \ref{sec:paramdependence} we observe that the policies we develop can depend on the unknown parameter $\theta$ and avoid this by averaging over a prior. To illustrate these methods, in Section \ref{subsec:pcr} we consider an example from biology, a Polymerase chain reaction (PCR) experiment where DNA template is grown in liquid substrate. The population dynamics are modeled in a dynamical system with stochastic errors, and the aim is to estimate the half-saturation constant, a parameter which controls the saturation of the template. Here we compare using a prior for $\theta$ and using VIA to calculate a control policy.

Throughout this paper we examine estimating only one parameter of interest. In many real-world scenarios, there will be multiple parameters that are the focus of attention. The techniques below can be readily extended to linear combinations of Fisher Informations for different parameters, including the trace of the Fisher Information matrix: T-Optimal designs. Other functions of the Fisher Information Matrix such as the determinant could also be targeted, although this would require a more substantial extension of the methods presented here (see \cite{ChalonerVerdinelli95} for an overview, and Section \ref{subsec:ext}). The methods we describe focus solely on one parameter at a time in order to maintain the connection between the design criterion and the asymptotic variance of the MLE. In more complex cases, design criteria potentially trade-off precision in one parameter in favor of another and must be carefully constructed for a specific problem. More complex design tasks for nonlinear dynamical systems is an important open problem, but beyond the scope of this paper.

Practical implementation of the methods in this paper will depend on the context. In some cases, inputs can be manually manipulated, but in systems which evolve in shorter time-scales recent advances in automating experiments may be needed. However, we draw a distinction between our problem of a single, evolving system with automatic protocols that perform multiple, parallel experiments (see \cite{king2009,hayden2014}). We note that the numerical implementation of our methods restricts us to systems low-dimensional state spaces and systems that are known up to the parameters we wish to estimate (see \cite{bongard2007} in contrast in ordinary differential equation models). Over-coming these obstacles represents an important focus of future research.


\section{Framework and Assumptions} \label{sec:framework}

We consider a Markov decision process \mbox{$(x_t, u_t)_{t=0,\ldots,T}$}. In this setting $x_t$ is a Markov chain, but the transition probabilities at time $t$ depend on a control $u_t$ chosen at that time. We assume a finite state space $\mathcal{X}$ for the state process $x_t $ and that the controls available belong to some finite set $\mathcal{U}$. We let $K$ denote the size of $\mathcal{X}$ and $l$ the size of $\mathcal{U}$. The transition probabilities are assumed to be parametric and we write $p(x_{t+1}|x_t,u_t,\theta)$ short for $p(x_{t+1} = x^i| x_t = x^j,u_t= u^r,\theta)$ where $x^i, x^j \in \mathcal{X}$ and $u^r\in \mathcal{U}$.

In addition to this, we assume that the process $x_t$ is latent and we only observe the related observations $y_t \in \mathcal{Y}$ whose relation to the $x_t$ can also depend on $\theta$. We write $p(y_t|x_t,\theta)$ short for $p(y_t=y^i|x_t=x^j,\theta)$, where $x^j\in \mathcal{X}$ and $y^j \in \mathcal{Y}$, and let $L$ denote the size of $\mathcal{Y}$.
This makes the system a Partially Observed Markov Decision Process.
It has a finite horizon $T$ in which we observe $y_0 \ldots y_T$. We will use the short hand notation $y_{m:t}$ to denote $y_m,\ldots,y_t$, i.e. the observations between time $m$ and $t$, and analogous notation for $u_t$ and $x_t$.

The objective is to use the controls to maximize the information we get about the parameter $\theta$ through the observed process $y_{0:T}$. The parameter estimation is done using maximum likelihood and it is therefore natural to try to maximize the Fisher Information of our observed process which we can express as
\[
FI=E\left[\sum_{t=0}^{T-1}\left(\frac{\partial}{\partial\theta} \log p(y_{t+1}| y_{0:t}, u_{0:t},\theta) \right)^2  \right ]
\]
Details on this construction   and regularity conditions required for its existence are given in  Appendix \ref{sec:express FI}. We have suppressed the dependence of this system on the initial state $x_0$. In our examples we treat $x_0$ as known, but it may also be marginalized with a simple modification of the filter below.
When we consider continuous time dynamical systems the observation spaces will be continuous, but we will use this discretized Fisher Information as an approximation to the actual Fisher Information of the observations.

\subsection{A Dynamic Program}\label{sec:dynprog}{

In order to maximize the Fisher Information we employ the techniques of stochastic dynamic programming. In the context of a (completely observed) Markov Decision Process, we consider obtaining a reward at each time point $t$ that is equal to $C_t(x_t,u_t)$. Our objective is to maximize the total expected reward $E[\sum_{t=0}^{T-1} C_t(x_t,u_t)]$ by use of the controls. The essence of dynamic programming is that by starting at time $T-1$ we can determine the best action $u_{T-1}(x_{T-1})$ to be taken for each value of the state. Working backwards, we can compute an optimal policy $u_t(x_t)$ that maps a state $x_t$ to a control $u_t$ that will maximize the expected reward, accounting for the already-calculated policies $u_s(x_s)$ for $s > t$.

In a generic dynamic program we set $V_T= 0$ and then going backwards from $t=T-1, \ldots, 0$ solve
\[
V_t(x_t) = \max_{u_t} \{C_t(x_t,u_t) + E_{x_{t+1}}[V_{t+1}(x_{t+1})|x_t,u_t] \}
\]
where $V_t$ is called the value function, and we get the associated control
\[
u^*_t(x_t) = \argmax_{u_t} \{C_t(x_t,u_t) + E_{x_{t+1}}[V_{t+1}(x_{t+1})|x_t,u_t] \}
\]
for every state $x_t$. This will give us a policy of what control to use at a certain state $x_t$ at a certain time $t$. The use of this control policy will maximize the total expected reward $E[\sum_t C_t(x_t,u_t)]$. We refer to \cite{puter} for a detailed description of dynamic programming.

With this in mind, we return to the POMPD setting. To choose controls to maximize the Fisher Information, we set
\[
C_t(y_{0:t}, u_{0:t},\theta) =  \left(\frac{\partial}{\partial\theta} \log p(y_{t+1}|y_{0:t}, u_{0:t},  \theta)\right)^2
 \]
 and we try to maximize $FI(\theta)=E[\sum_t C_t(y_{0:t}, u_{0:t},\theta)]$. Note that in this instance the reward function depends on the entire history of observations and controls up to time $t$ which will motivate our approximation below.

The Value function in the corresponding dynamic program is
\[
FI_t(y_{0:t},u_{0:(t-1)}) = \max_{u_t} \left\{ E_{y_{t+1}}[C_t(y_{0:t},u_{0:t},\theta) +FI_{t+1}(y_{0:(t+1)},u_{0:t},\theta)|y_{0:t},u_{0:t},\theta ]   \right \}
\]
and we denote it the Fisher Information to Go .
%
%
%

\subsection{Partial Observation Fisher Information}

A problem with the formulation above is that just in the first step of the dynamic program ($t=T-1$) we would have to calculate the Fisher Information to Go for $L^{T}l^{T-1}$ many combinations of $y_{0:(t-1)}$ and $u_{0:(t-2)}$.
This is formidable for even modest dimensions. We therefore approximate the process by conditioning only on the last $m+1$ observations. We label the resulting approximation the ``Partial Observation Fisher Information'' to distinguish it from the ideal target and from the FOFI objective used in \cite{FOFI}:
\[
POFI_m = E\sum_{t=0}^{T-1}\left(\frac{\partial}{\partial \theta} \log p(y_{t+1}|y_{(t-m):t},u_{(t-m):t},\nu_{t-m},\theta)\right)^2
\]
where $\nu_{t-m}$ is some prior that we assume for $x_{t-m}$, although we generally suppress it in notation since we assume it is fixed. If $t-m <0$ we set $(t-m):t$ to mean $0:t$ to ease notation. We will use ``POFI'' in place of $POFI_m$ when discussing the general strategy and note that even $m=1$ produces useful designs.

The reward becomes
\[
C(y_{(t-m):t},u_{(t-m):t} ,\theta)= \left(\frac{\partial}{\partial\theta} \log p(y_{t+1}|y_{(t-m):t},u_{(t-m):t} ,\theta)\right)^2
\]
and
 \[
 POFI_{t,m} (y_{(t-m):t},u_{(t-m):(t-1)}) = \max_{u_t} \left\{E_{y_{t+1}}[C(y_{(t-m):t},u_{(t-m):t} ,\theta)+FI_{t+1,m}|y_{(t-m):t},u_{(t-m):t} ,\theta) ]   \right \}
 \]
the Partial Observation Fisher Information To Go. The pseudocode for the corresponding dynamic program is given in Appendix \ref{sec:dyn POFI}.
%
For this approximate dynamic program to be sensible we want the approximated Fisher Information to approach the true Fisher Information as $m$ increases. This holds given that the POMDP process satisfies certain technical mixing conditions:
\begin{myassump}
 Modified Strong Mixing Conditions. For each control $u$ there exist a transition kernel $K^u: \mathcal{Y}\rightarrow \sigma(\mathcal{X})$ and measurable functions $\varsigma^-$ and $\varsigma^+$ from $\mathcal{Y}$ to $(0,\infty)$ such that for any $A \in \sigma(\mathcal{X})$, $y \in \mathcal{Y}$ and $x\in \mathcal{X}$,
 \[
  \varsigma^-(y)K^u(y,A) \leq \sum_{x'\in A} p(y_{t+1}=y|x_{t+1}=x')p(x_{t+1}=x'|x_t=x,u_t=u) \leq \varsigma^+(y) K^u(y,A)
 \]
 where $\sigma(\mathcal{X})$ signifies the $\sigma$-algebra of $\mathcal{X}$.
 \label{mixing}
\end{myassump}
These conditions bound the probability of, starting at $x$, observing $y_{t+1}=y$ via a transition through the set $A$. These are a generalization of conditions given in \cite{Cappe} to POMDPs and the discussion of the systems that satisfy Assumption \ref{mixing} given there generalizes readily.

Employing this assumption, we can show the following result:
\begin{theorem}
Assume the conditions in Assumption \ref{mixing} hold at $\theta$. Then, for $m < T$ and any control policy, we have
\[
 |FI - POFI_m| \leq c_1(T-1-m)\rho^{m/2}
 \]
\label{FI approx}
where the constant $c_1$ and $\rho<1$ do not depend on $m$ or $T$.
\end{theorem}
Explicit expressions for $c_1$ and $\rho$ are given in Appendices \ref{sec:bounds} and \ref{sec:mixing} respectively.
The proof requires extensions of work in \cite{Cappe} and is given in Appendix \ref{sec:bounds} with further technical results reserved to the Appendix \ref{app:B}.
Theorem \ref{FI approx} states that the difference between $POFI_m$ and the ideal Fisher Information decreases exponentially quickly as $m$ increases. Thus $POFI_m$ represents a viable approximation for the Fisher Information in the dynamic program when $m$ is small. Below, we demonstrate good performance even for $m=1$.

The runtime of the dynamic program however also grows exponentially in $m$ and we found that while setting $m=0$, i.e. conditioning on one observation, gave poor results in some of our simulations, conditioning on two observations, i.e. $m=1$, generally gave good results when compared to other control policies. In experiments with discrete systems, $m=2$ and $m=3$ generated further improvements but with clearly diminishing returns. For systems with large numbers of states or when discretizing continuous systems, setting $m=2$ increased runtime greatly and was in most of our applications infeasible without making more approximations to how the dynamic program is run. 
We leave a problem-specific analysis of a means to choose $m$ to future work, but note that this can, at a minimum, be approached by simulation.

\subsection{FOFI Dynamic Program}
An alternative method to choose controls was proposed by \cite{FOFI}. They considered constructing an optimal control policy for the Fisher Information that would apply if $(x_t)$ were observed directly;
\[
FI =E\sum_{t=0}^{T-1}\left(\frac{\partial}{\partial \theta} \log p(x_{t+1}|x_t,u_t,\theta)\right)^2
\]
This is labeled the Full Observation Fisher Information (FOFI). As noted before, when considering continuous time stochastic systems, the state space is continuous, but we use this Fisher Information as an approximation to the continuous state Fisher Information. An advantage of using FOFI over POFI is that when running the dynamic program the Markov property of the Markov Decision Process $(x_t,u_t)$ allows us to only consider a maximization over the state space $x_t \in \mathcal{X}$ but not past values $x_{0:(t-1)}$. The dynamic program for FOFI is given in Appendix \ref{sec:dyn FOFI}.

%

However, maximizing FOFI can lead to suboptimal controls since it is not the correct Fisher Information for the data. Additionally, when the actual experiment is run we do not observe $x_t$. Instead we have to use the observed values to get a probability distribution (a filter) on the state $x_t$, $p(x_t|y_{0:t},u_{0:(t-1)},\theta)$ and use the control associated with the state that has the highest probability.

The computation cost of running a dynamic program with FOFI is $O(TK^2l)$, generally lower than that of POFI; $O(TL^{m+1}l^{m+1})$. The cost of estimating the state $x_t$ at runtime is $O(K^2)$ at each time point $t$. For details see Appendix \ref{app:B1}. The exponential increase in cost as $m$ increases forces us to choose a small value of $m$ in our experiments below, which may make the POFI approximation to the Fisher Information suspect, but we found that even at these values, our controls yielded improved parameter estimates.
%
%
%

\subsection{Parameter Estimation}
After running an experiment, using one of the control policies, the parameter $\theta$ is estimated either via an EM algorithm or by directly maximizing the loglikelihood. These two estimation methods had similar accuracy in our simulations, although the EM algorithm was generally slower. Convergence of the EM algorithm is discussed in  \cite {Cappe} for Hidden Markov Models and extends naturally to POMDP's.
%

For the asymptotic properties of the MLE we refer to \cite{Cappe}, where conditions for consistency and asymptotic normality in Hidden Markov Models are given. The central elements of their proof are the stationarity of the process $(x_t,y_t)$ along with certain forgetting properties of the filter, meaning that ignoring all but the past $m$ observations (as we do) yields an error in the filter distribution that decreases exponentially in $m$. We note that if we employ a time-independent control policy (as we do in Section 5), we obtain a Hidden Markov Model and can rely on \cite{Cappe} if we assume stationarity. In the Appendix \ref{app:B} we establish extensions of forgetting properties for POMDP models more generally, which points to a more general asymptotic theory for the MLE in this case, but do not pursue this here.

Theorem \ref{FI approx} shows that using $POFI_m$ is a good approximation to the Fisher Information for running a dynamic program. This provides a control policy that is an approximation to the optimal control policy. Now consider using this approximate policy to run an experiment and then estimating $\theta$ by evaluating the MLE. The asymptotic variance of this MLE will be the  inverse of the Fisher Information, with controls from the approximate policy. It is therefore of interest to compare the Fisher Information with an optimal policy and that with the approximate policy derived from the POFI objective. In Appendix \ref{sec:best FI} we show
\begin{theorem} \label{best FI}
Given that the mixing conditions in Assumption \ref{mixing} (\ref{sec:mixing}) hold we have
\[
 0 \leq FI(u_0^*,\ldots,u_{T-1}^*) - FI(u_{0,m}^*,\ldots, u_{T-1,m}^*) \leq c_2T(T+1)  \rho^{m/2}
\]
where $u_0^*,\ldots,u_{T-1}^*$ is the optimal control policy, $u_{0,m}^*,\ldots, u_{T-1,m}^*$ the approximated optimal policy and $FI(u_0^*,\ldots,u_T^*)$, is the Fisher Information written as a function of the policy. The constant $c_2$ and $\rho$ do not depend on $m$ or $T$, see Appendix \ref{sec:best FI} for details.
\end{theorem}

 That the asymptotic variance of the MLE converges to the best possible asymptotic variance exponentially quickly in $m$ further supports our approximations.
%

\subsection{Extensions to Other Design Criteria} \label{subsec:ext}

Here we briefly comment on potential extensions to other experimental design criteria, although a full exploration of these is left to future papers. The general dynamic programming framework can readily be applied to a criteria of the form $E D(y_1,u_1,\ldots,y_{T},u_T)$ by describing the policy at time $t$ as
\[
u^*_t (y_1,u_1,\ldots,y_{t-1},u_{t-1}) = \mbox{argmax} E_{y_{t:T},u^*{(t+1):T}} D(y_1,u_1,\ldots,y_{T},u_T).
\]
However, this will be infeasible unless a lower dimensional approximation to $u^*_t$ can be found.

A further set of criteria involve functions of the Fisher Information matrix. We have already noted that the extension to weighted sums of the diagonals is straightforward. Alternatively, the asymptotic variance of an individual parameter is given by the corresponding entry of the inverse of the Fisher Information. We can target this -- effectively treating the remaining parameters as nuisance variables -- by setting
\[
u^*_t (y_1,u_1,\ldots,y_{t-1},u_{t-1}) = \mbox{argmin} \left[ \left( E_{y_{t:T},u^*{(t+1):T}} \frac{\partial^2}{\partial \theta \partial \theta^T} \log p(y_1,u_1,\ldots,y_{T},u_T) \right)^{-1} \right]_{11}.
\]
We expect that our results in Theorem \ref{FI approx} on truncation approaches to targeting individual entries of the Fisher Information can be extended to functions of the whole matrix. Similarly, we can apply the same ideas within the Value Iteration Algorithm given in Section \ref{sec:via} to account for parameter dependence, but leave these problems for future work.

\section{Discrete Examples} \label{sec:discreteexamples}
\subsection{3 state example}
%

While the FOFI strategy has been shown to be effective in \cite{FOFI} it is possible to define systems in which the strategy is not optimal and may in fact be worse than just using fixed or random controls. Usually certain parts of state space will give more information about a parameter than others, given that the state space is perfectly observed. In these cases optimal controls would try to move the process to these states. However, if the state space is only partially observed, most information might be obtained in different parts of state space and the FOFI controls become suboptimal. In cases like this, POFI often does better than FOFI even truncating to $m=1$. In this example, we demonstrate a system where FOFI and POFI choose very different controls, and using a simulation study, we show that the controls chosen by POFI produce less variable parameter estimates.

Consider a discrete time Markov chain $x_t$ with state space $\mathcal{X}=\{1,2,3\}$ and a transition probability matrix
\[
P=
\left[
\begin{array}{c c c}
1/2-p/4+u/4 & 1/3 & 0.4 - u/4 \\
p/2 & 1/3 & 0.15\\
1/2-p/4-u/4 & 1/3 & 0.45 + u/4\\
\end{array}
\right]
\]
where the parameter of interest is $p \in [0,.5]$ (this range is chosen to maintain positive entries in $P$) and the control is $u \in \{-1,1\}$. For $x_t=1$ or $x_t=3$, choosing the control $u=1$ will increase the probability of the Markov chain staying in its current state while choosing $u=-1$ will increase the probability of it leaving its state.\\

Also, assume there is a related process $z_t$ with state space $\mathcal{Z}=\{1,2\}$ whose transition probabilities depend on which state $x_t$ is in, and out of the two processes only $z_t$ is observed, that is we have observations $y_t = z_t$. We denote the transition probability matrices for $z_t$ with $(P_k)_{\{i,j\}}=p(z_{t+1}=j|z_t=i,x_{t+1}=k)$ given by
\[
P_1=
\left[
\begin{array}{c c}
.5& .5 \\
.5 & .5\\
\end{array}
\right],
P_2=
\left[
\begin{array}{c c}
.5 & .5 \\
.5 & .5\\
\end{array}
\right],
P_3=
\left[
\begin{array}{c c}
1-p/2 & p/2 \\
p/2 &1-p/2 \\
\end{array}
\right]
\]
If $x_t$ were observed we would get information about the parameter $p$ when $x_t$ leaves state $1$ and from $z_t$ when $x_t=3$. The idea here is that since the FOFI controls assume the whole state space is observed they might encourage $x_t$ to be in state $1$, while the POFI controls that take into account what is actually observed might choose the controls more intelligently.
Indeed when calculating the controls according to FOFI, the long run control is to ``leave one's state'' if $x_t=3$ and ``stay in one's state'' if $x_t=1$. The POFI policy when only $z_t$ is available ($m=0$) is to always have $u_t=-1$. The policy for POFI with $m=1$ is given in Table \ref{6 state POFI lag 2} and can be summarized as ``$u_t =1$ if $z_t = z_{t-1}$, and otherwise $u_t=-1$''. The policy for POFI with $m=2$ is given in Table  \ref{6 state POFI lag 3} in the Supplementary Materials, but can be summarized as  `` if $z_t = z_{t-1}$ and ($z_{t-1} = z_{t-2}$ or $u_{t-1} = -1$) then $u_t = 1$, else $u_t = -1$''.



\begin{table}
\begin{center}
\begin{tabular} {|c || c  c  c c c c c c |}
\hline
$u_t$  & 1& -1&-1 &1& 1&-1&-1&1\\

\hline
\hline
$z_t$       &1 & 2  &  1 &  2  & 1 & 2  & 1  &  2 \\
$u_{t-1}$ &1 &1  & 1 &  1  & -1 &-1 & -1 & -1  \\
$z_{t-1}$ &1 &1  & 2  &  2  & 1 & 1  &  2  &  2 \\

\hline
\end{tabular}

\end{center}
\caption{Long run control policy that results from using POFI, with $m=1$ in the 6 state example. The first row describes which control to use for a given history $(z_t,u_{t-1},z_{t-1})$ of observations and control. We see that $u_t =1$ if $z_t = z_{t-1}$, and otherwise $u_t=-1$ }
\label{6 state POFI lag 2}
\end{table}

\begin{table}
\begin{center}
\begin{tabular} {|c| c c c c|}
\hline
 & $m$ & bias & st. dev. & RMSE \\ 
 \hline
 FOFI & - &  0.0025 & 0.0798 & 0.08 \\ 
 POFI & 0 &   0.0037 &  0.0591 & 0.059 \\ 
  POFI & 1 &  0.0018  &0.0469  & 0.047 \\ 
 POFI & 2 &  0.0054 & 0.0467 &  0.047 \\ 
 Random & - & 0.0013 &  0.0621 & 0.062 \\ 
 \hline
\end{tabular}
\end{center}
\caption{Simulation results for the $6$ state example. We see that the controls chosen by POFI make for more accurate estimates of $p$. The FOFI policy does worse than a random policy.}
\label{6 state result}
\end{table}

To illustrate this difference, a simulation study was carried out to test which method performed better: the process $x_t$ was run for $1000$ steps with $p=.37$, using controls chosen by POFI and again using those chosen by FOFI. Additionally we ran a simulation of the same length, but where the control was chosen randomly, with $u=-1$ and $u=1$ having equal probability. Then the parameter $p$ was estimated using an EM algorithm. This was done 500 times to get an empirical distribution for the estimates of $p$. The results are given on the right in Table~\ref{6 state result}. Estimates of $p$ improved with the number of lags included in the POFI policy, but greatest improvement was observed from $m=0$ to $m=1$. Even at $m=0$, POFI provided an improvement on the Random policy while FOFI under-performed even selecting controls at random.

\subsection{Adversarial game - A POMDP example}
The following example describes an application of the above methods in the context of experimental economics. The problem is derived from  \cite{adv}, in which we wish to model how humans change their game-playing strategies over time.
%

We set up a game with two players: a Row player and a Column player. They repeatedly play a game where both simultaneously choose either left or right, and they get rewards depending on the outcome, according to Table \ref{gamble}; the Row player would for example get $2$ and the Column player $0$ if both chose left.
We follow \cite{adv} and assume that at any given play the Column player follows one of two strategies: the Nash-equilibrium strategy of choosing either left or right with $50 \%$ probability or the Gamble-safe strategy, where they only choose right. The player will pick either strategy based on a multinomial logistic model, where the probabilities depend on the last two plays of the Row player, and the last strategy chosen by the Column player.
This results in a Partially Observed Markov Decision Process with the strategy employed being a hidden state giving rise to observed plays.

\begin{table}[!]
\begin{center}
\begin{tabular}{|c  |c| c|}
\hline
& Left & Right \\
\hline
Left  & 2,0 & 0,1\\
\hline
Right & 1,2 & 1,1\\
\hline
\end{tabular}
\end{center}
\caption{Rewards in the Gamble Safe game. The first number is the reward for the Row player and the second number the reward for the Column player, given a certain outcome. }
\label{gamble}
\end{table}

Let $s_t$ denote the strategy chosen by the Column player at time $t$ and $u_t$ denote the action played by the Row player at time $t$. Let $s_t=1$ if the Nash-equlibrium is chosen, $s_t=-1$ if the Gamble-safe strategy is chosen. Also let $u_t=1$ if the Row player plays right, $u_t=-1$ if he plays left. Similarly $y_t$ will denote the plays of the Column player. The strategy $s_{t+1}$ chosen at time $t+1$ will then be chosen according to
\[
P(s_{t+1}=1)= \frac{e^{x_t}}{1+e^{x_t}} \textrm{   and  } P(s_{t+1}=-1)= \frac{1}{1+e^{x_t}}
\]
where we let
\[
x_t=1.2u_t+u_{t-1}+\theta s_t
\]

The experiment is set up with two natural strategies for the Column player and we can think of $\theta$ as the persistence of strategies. The purpose of this experiment is to elicit information about how humans persist in strategy choice, and we therefore investigate how the plays of the Row player can be used to obtain an estimate of $\theta$ that is as precise as possible.

To cast this into a POMDP setting we think of $s_t$ being the unobserved underlying Markov Chain, $u_t$ as the control and $y_t$ as the observed process. Since the transition probabilities from $s_t$ depend on $u_{t-1}$ (a part of the history at time $t-1$) we augment the state space to include $u_{t-1}$, i.e. $r_t=(s_t,u_{t-1})$ will be our underlying Markov Chain. At this point we could run the dynamic programs for both FOFI and POFI, but controls calculated that way will depend deterministically on the plays of the Column player. Seeing that realistically deterministic plays can often easily be countered in adversarial games, it is better to follow a strategy that includes some randomness in the plays. So we let $w_t \in \{-1,1\}$ be the strategy of the Row player in such a way that

\[
\left.
\begin{array}{l l}
u_t = 1 & \text{w.p. $.8$} \\
u_t = -1&  \text{w.p. $.2$} \\
\end{array}\right\} \textrm{ if } w_t=1, \quad \textrm{ and } \quad
\left.
\begin{array}{l l}
u_t = 1 & \text{w.p. $.2$} \\
u_t = -1&  \text{w.p. $.8$} \\
\end{array}\right\} \textrm{ if } w_t=-1
\]
These kind of changes are easily incorporated in the dynamic program for both POFI and FOFI, by adding an expectation over $w_t$ at every step $t$.
%
%
%
%

We set $\theta=.7$ and calculated the FOFI and POFI policies with lags up to $m=7$. The long run FOFI policy was to set $w_t = -1$ if $u_{t-1}=1$ and $w_t = 1$ if $u_{t-1} = -1$ (no matter the state $s_t$) which indicates a preference for alternating the control at every step. Thus if $w_t=1$ is chosen at time $t$ and $u_t=1$ is sampled, then at time $t+1$ we set $w_{t+1}=-1$. The long run POFI policy with $m=1$ was to choose $w_t = 1$ if $y_t=1$ and $u_{t-1}=u_{t-2}=-1$ and $w_t = -1$ otherwise; policies with more lags were too complicated to list.

We ran simulation studies with $T=500$, and $1000$ simulations to compare the POFI and the FOFI policies. Another $1000$ simulations were run where the plays (control) where chosen randomly. The parameter $\theta$ was estimated using an EM algorithm. The results of this estimation under each policy are given in Table~\ref{adversarial results} where the POFI controls with $m>1$ produced the most stable results, but little additional improvement is seen for $m>3$. Table~\ref{adversarial results2} in the Supplementary Materials provides results for up to $m=7$.

\begin{table}
\begin{center}
\begin{tabular} {|c| c c c c|}
\hline
 & $m$ & bias & st. dev. & RMSE \\ 
 \hline
 FOFI & - & 0.0051  & 0.2675 & 0.2676 \\ 
 POFI & 0 & 0.0115 & 0.2717 & 0.2718 \\ 
  POFI & 1 & 0.0022 & 0.2712 & 0.2711 \\ 
 POFI & 2 & 0.0096 & 0.2565 & 0.2567 \\ 
  POFI & 3 & 0.0100 & 0.2489 & 0.25 \\ 
 Random & - & 0.0075 & 0.2667 & 0.2668 \\ 
 \hline
\end{tabular}
\end{center}
\caption{Simulation results for the adversarial game. The POFI control policy, with $m>1$, offers the best controls for optimizing the estimation of $\theta$ with MSE at or under $.6$}
\label{adversarial results}
\end{table}


%
%
%

 \section{Discretization methods} \label{sec:discretization}
 In order to apply the methods described above to more general dynamical systems, we need to approximate them by a suitable Partially Observed Markov Decision Process. We achieve this by discretizing time, state and observation spaces.
 In this paper, the continuous stochastic dynamical systems considered are of the form
\[
	d \xbold = \fbold(\xbold,\theta,u(t)) dt + \Sigma_1^{1/2} d\mathbf W
\]
where $\theta$ is the parameter of interest, to be estimated, $u(t)$ is a control that can be chosen by the user, $\xbold$ is the vector of state variables, $\fbold$ is a vector valued function and $\mathbf W$ a Wiener process.
The dynamical system is approximated on a fine grid of times $(t\delta)_{t=0,\ldots,T}$ and we obtain a discrete-time model
\[
\xbold_{t+1} = \xbold_t + \delta \fbold(\xbold_t,\theta,u_t) + \sqrt{\delta}
\epsilonbold_{1t}
\]
where $\epsilonbold_{1t} \sim N(0,\Sigma_1)$ are independent normal random variables. We assume the underlying state variables $x_t$ are only observed partially or noisily:
\[
\ybold_t = \gbold(\xbold_t) + \epsilonbold_{2t}
\]
where $\epsilonbold_{2t} \sim N(0,\Sigma_2)$.

In order to approximate this as a Markov Chain, the state space is discretized in each dimension and the model is then thought of as moving between the different boxes. The probability of moving from box to box is approximated using the normal p.d.f. at the midpoints of the boxes.
 In the examples in this paper, only equidistant discretization is considered, but this restriction can be readily removed. If we label the two midpoints as $i_1$ and $i_2$ and the area of the second box as $A_x$ this probability is given as
\begin{align*}
& p(x_{t+1} = i_2| x_t=i_1,u_t,\theta) \\
& \hspace{2cm} \propto \frac{\exp \left (-\frac{1}{2}(i_2-(i_1+\delta \fbold(i_1,\theta,u_t)))^T \Sigma_1^{-1} (i_2-(i_1+\delta \fbold(i_1,\theta,u_t)))\right)\cdot A_x}{(2\pi)^{k/2} \det(\Sigma_1)^{1/2}}
\end{align*}
where $k$ is the dimension of $\xbold$.
The probabilities are then normalized to make sure they sum to $1$.
If the controls $u_t$ can be chosen on a continuous scale then this scale has to be discretized as well. $(x_t,u_t)$ is then a Markov Decision Process, and one can run the FOFI dynamic program.

For the POFI dynamic program the observation space needs to be discretized as well. The probability of what observation box is observed depends on in which box the underlying Markov Chain is in. If we label the midpoint of the underlying Markov chain midpoint as $i$ and the midpoint of the observed process box midpoint as $j$, and the area of the latter box as $A_y$ this probability is given as
\[
p(y_t=j|x_t=i) \approx\frac{1}{(2\pi)^{k/2} \det(\Sigma_2)^{1/2}} \exp\left(-\frac{1}{2}(j-g(i))^T \Sigma_2^{-1} (j-g(i))\right)\cdot A_y
\]
These probabilities are also normalized to sum to $1$.
The process $(x_t,y_t,u_t)$ is now a Partially Observed Markov Decision Process and one can run the POFI dynamic program.

\subsection{Morris-Lecar Model} \label{subsec:ml}
The Morris-Lecar Model \cite{morris} describes oscillatory electric behavior in a single neural cell, as regulated by flow of Potassium and Calcium ions across the cell membrane.  The model is defined in terms of a voltage $v_t$ across the axon membrane and a gating variable $n_t$ that describes the fraction of Potassium channels that are open. A differential equation for these models is expressed as
\begin{align}
  C_m \dot{v}_t &= I_t - g_l\cdot(v_t-E_l) - g_K\cdot n_t\cdot(v_t-E_K) -
  g_{Ca}\cdot m_\infty(v_t)\cdot(v_t-E_{Ca}) \label{eq:ml-v}\\[2ex]
  \dot{n}_t &= -\phi\cdot (n_t - n_\infty(v_t)) / \tau_n(v_t) \label{eq:ml-n}
\end{align}
given in terms of auxiliary functions $m_\infty(v) = \frac12(1+\tanh((v-v_1)/v_2))$, $\tau_n(v) = {\rm sech}((v-v_3)/(2v_4))$ and $n_\infty(v) = \frac12(1+\tanh((v-v_3)/v_4))$.
We will write $C_m \dot{v}_t  = F_1(v_t,n_t)$ and $\dot{n}_t = F_2(v_t,n_t)$ as shortcuts equations (\ref{eq:ml-v}) and (\ref{eq:ml-n}).

In (\ref{eq:ml-v}-\eqref{eq:ml-n}), the first equation describes Kirchoff's current conservation law in which a current $I_t$ is injected into the neuron and the remaining terms represent the ``reversal potentials'' of each of a leakage current, Potassium and Calcium ions towards their equilibrium values respectively given by  $E_l, E_K$ and $E_{Ca}$. The conductance, $g_{Ca}$ of the Calcium channel is modified by a function of the current voltage $m_\infty(v)$, becoming more conductive (hence exerting a stronger influence on $v_t$) as $v$ increases. In contrast, the conductance of the Potassium channel changes dynamically as $n_t$ -- representing the number of open channels -- converges to its voltage-dependent equilibrium value of $n_\infty(v)$ more slowly than the rate of change of of $v$. For our purposes, we will be interested in the conductance, $g_{C_a}$, the relative speed, $C_m$, of the $v_t$ and $n_t$ variables, and $\phi$ -- the absolute speed of the $n_t$ dynamics. The remaining conductances $g_l$ and $g_K$ as well as the equilibria $E_l, E_k$ and $E_{Ca}$ could also be estimated, but were found in \cite{FOFI} to yield less interesting control policies.


In this model, when $I_t$ is large enough the neuron will ``spike'': producing rapid peaks in voltage that  stimulate connected neurons. A typical experiment involves applying a constant voltage and observing the neuron behavior. Here, we examine using $I_t$ as a control variable dynamically in order to maximize information about the parameters $C_m, g_{Ca}$ and $\phi$. We will examine experiments designed to target each parameter in turn because this is revealing about the sources of information for them, but a combined criterion such as the sum of their approximated Fisher Informations could readily be employed.

We consider a stochastic version of this neural firing model, derived from \cite{smith}, by adding $\sigma d w_1$ and $\tilde\sigma d w_2$ to equations (\ref{eq:ml-v}) and (\ref{eq:ml-n}) respectively, where $w_1$ and $w_2$ are independent Wiener processes.
 Stochastic models are important in this context in order to accommodate observable variation in the inter-spike interval where a deterministic model will require a fixed period;  see \cite{Hooker09}, for example.

The first step is to discretize these equations with respect to time. For a discrete time-step of size $dt$ we approximate $v_{t+dt} =v_t+dt \cdot F_1(v_t,n_t)/C_m + \sigma\sqrt{dt}\cdot\varepsilon_1$ and $  n_{t+dt} =n_t+dt\cdot F_2(v_t,n_t) + \tilde\sigma\sqrt{dt}\cdot\varepsilon_2$ where $\varepsilon_1,\varepsilon_2  \sim N(0,1)$.
%

We discretized $v_t$ onto the range $[-75,45]$ and $n_t$ onto $[0,1]$, after running a few trial versions of the model. Both ranges where discretized into $25$ intervals.
Only $v_t$ is measured and it is measured noisily,
$
 y_t=v_t +\varepsilon_{t},
$
where $\varepsilon_{t} \sim N(0,1)$. The observation space was discretized to the same range as $v_t$ but into $20$ intervals.
These approximations give rise to a Partially Observed Markov Decision Process to which our methods can be applied. FOFI and POFI (with $m=1$) controls were calculated for this experiment, targeting each of $C_m$, $g_{Ca}$ and $\phi$ in turn. The values for the parameters were set to be $C_m=20$, $g_{Ca}=4.4$, $g_l=2.0$, $E_k=-84.0$, $E_l=-60$, $E_{Ca}=120.0$, $\phi=.04$, $v_1=-1.2$, $v_2=18.0$, $v_3=2.0$, $v_4=30.0$, $\sigma =\tilde{\sigma}=1$ and $dt=1$. The controls range was set to be $[-1.5,6.0]$ and discretized to the set $I_t \in \{-1.5,0.0,1.5,3.0,4.5,6.0\}$. An example of  the control policy for $g_{Ca}$ is given in Figure~\ref{fig:morris-lecar} where we note that while the FOFI policy is described by the state variables $(v_t,n_t)$, the POFI policy is given in terms of observations $(y_t,y_{t-1})$.
A simulation study was run for each of the three parameters; the system was simulated within the discretized Markov Chain framework with $100$ time steps and all schemes had 100 simulations. The parameter in question was estimated for each simulation using an EM algorithm. As a baseline comparison we also ran a simulation study using a fixed control $I_t=1.5$ corresponding to experimental protocols typically used in practice. The results are given in Table~\ref{morris-lecar}. The difference between POFI and FOFI turns out to be not very dramatic, likely due to the observations providing a great deal of information about the underlying state variables; the scenario in which  FOFI performs well; see \cite{FOFI}. Extending this simulation to more than $m=1$ incurred significant computational costs; \cite{FOFI} reported little difference between estimation results for fully observed systems (for which FOFI is the Fisher Information) and where $v_t$ is observed noisily, suggesting that there is little additional information to be gained from further lags.

\begin{table}[!]
\begin{center}
\begin{tabular} {|c| c c c c|}
\hline
 parameter& & bias & st. dev. & RMSE \\ 
 \hline
 \multirow{3}{*}{$C_m$}  & FOFI &  .4234 &2.4722  & 2.5086 \\
 & POFI &.4129   &2.4068  &2.442 \\ 
 & Fixed & .9098& 3.4240  & 3.5427 \\ 
\hline
  \multirow{3}{*}{$g_{Ca}$}  & FOFI &.0613   & .3671  & .3722 \\ 
 & POFI &.0158   &.3706  & .3709 \\ 
 & Fixed & .0249 & .6193 & .62 \\ 
\hline
  \multirow{3}{*}{$\phi$} & FOFI &  .00485 & .01085 & .01183 \\ 
 & POFI &.00257   &.01037  & .0105 \\ 
 & Fixed & .01357& .02643 & .03 \\ 
 \hline
 \end{tabular}
 \caption{Simulation results for the Morris-Lecar model, consider the parameters $C_m, g_{Ca}, \phi$ separately. We see that the POFI and FOFI policies outperform the fixed policy $I_t=1.5$ in all cases, and the POFI policy seems to perform slightly better than the FOFI policy for the three parameters considered. }
 \label{morris-lecar}
 \end{center}
\end{table}

%
%

\begin{figure}
  \centering
  \subfloat[FOFI control]{\label{fig:gull}\includegraphics[scale=0.2]{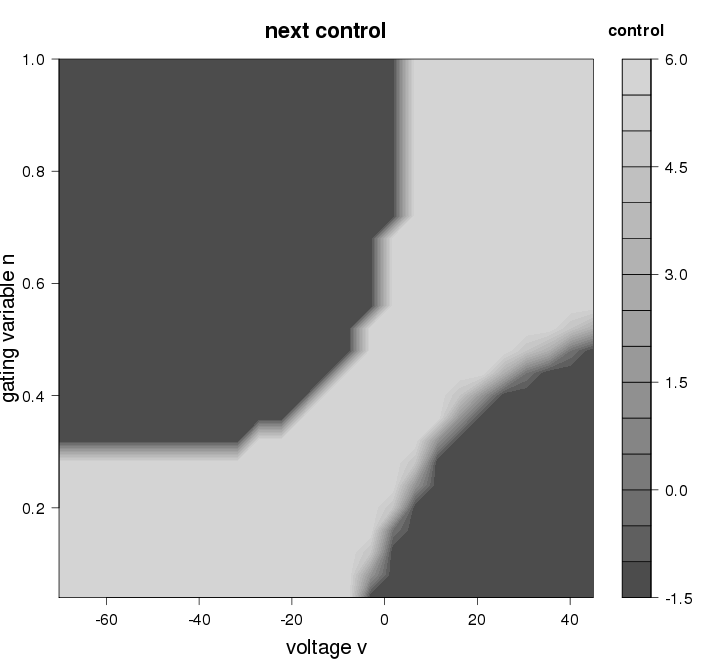}}
  \subfloat[POFI control]{\label{fig:tiger}\includegraphics[scale=0.2]{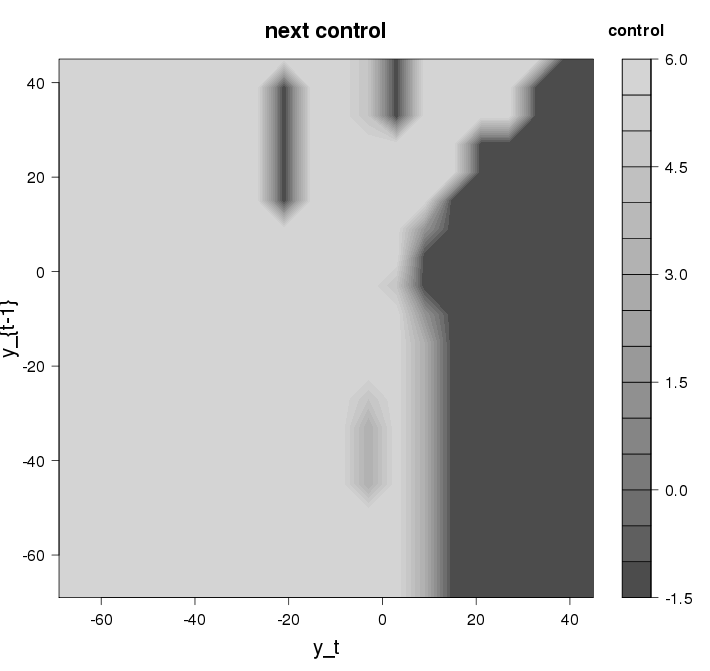}}
  \caption{Long term controls of FOFI and POFI for the parameter $g_{Ca}$. The FOFI plot gives the control to use, given a certain position in state space. The POFI control will depend on the last two observations and the last control, but fixing the last control as, for example, $I_{t-1}=6$ one can plot which control to use given combinations of the last two observations. }
  \label{fig:morris-lecar}
\end{figure}

\section{Parameter dependence of dynamic program} \label{sec:paramdependence}
In the discussion above we calculated the dynamic program assuming knowledge of the parameter $\theta$, the very thing we wish to estimate with maximal precision. Since the dynamic programs we have considered are run before the experiment is started, we generally won't have data to estimate $\theta$. Additionally, for the FOFI simulations we have used $\theta$ directly to estimate $x_t$ within the filter to get the appropriate control, but this will not be possible in practice.
There are a few ways of dealing with this.

Assuming some prior information one can use a prior for $\theta$ to run the dynamic program. To do this, we add one more expectation for $\theta$ at every time step $t$, and then maximize the expected Fisher Information to get the best control.
This strategy was employed in \cite{FOFI}. Experiments with controls derived from this policy for both POFI and FOFI are reported in Table~\ref{chemostat} along with those for the methods we propose below.

The rather obvious deficiency with averaging over a prior, for either POFI or FOFI, is that as the experiment runs, we get observations that can be used to improve our prior for $\theta$, and could be used to get better controls, if we could brake the experiment and rerun the dynamic program.

\subsection{Online updating}
In some systems the time spent in each state is very short, too short to perform many calculations, making it valuable to have a ``look-up table" of controls. Here the POFI controls have an advantage over the FOFI controls, in the sense that they are of the ``look-up" kind, as FOFI requires estimation of the underlying $x_t$ process, before the control can be looked up.

In other systems, there is time to do some calculations between transitions. Note, for example, that at time $t$ we have observed $y_0,\ldots, y_t$  and this will allow us to calculate a posterior distribution $\pi(\theta|y_{0:t}, u_{0:(t-1)})$ for our parameter of interest. This posterior could then be used to run the dynamic program again, as described above, from time $T-1$ to time $t$. This can be quite time consuming if done at each time step $t$, so we propose a method that relies on the Value Iteration Algorithm.

\subsection{Value Iteration Algorithm} \label{sec:via}
A popular algorithm from the theory of Markov Decision processes is the Value Iteration Algorithm (VIA), see  \cite{puter}. The theoretical motivation of VIA is similar to dynamic programming, but here the objective is to maximize an expected total reward $W_1$ that has a discounting factor $\lambda$, where $0 \leq \lambda < 1$, and the time horizon is assumed to be infinite;
\[
W_1=E\left [ \sum_{t=0}^{\infty} \lambda^{t-1} C(x_t,u_t)\right]
\]
and $W_1$ is labeled as the expected total discounted reward. $W_1$ exists if $C$ is bounded, which is the case in the problems we consider.
In \cite{puter} it is shown that an optimal control policy exists and it can be chosen to be time independent, i.e. to depend only on the state $x_t$, and not the time $t$. Moreover this optimal control can be approximated using the Value Iteration Algorithm described below. Our experimental setting is neither discounted nor has it an infinite time horizon, but Blackwell optimality guarantees that  controls that maximize $W_1$ also maximize the expected average reward $W_2$ (or its $\limsup$ if the limit doesn't exist);
\[
W_2 = \lim_{n\rightarrow \infty} \frac{1}{n} E\left[ \sum_{t=0}^{n} C(x_t,u_t)\right]
\]
given that $\lambda$ is chosen close enough to one. A Blackwell optimal control policy exists if the state and action spaces are finite, which is the case in our setting. Maximizing $W_2$ effectively amounts to maximizing the average input of each observation in our Fisher Information; a reasonable strategy. How small $1-\lambda$ needs to be is generally hard to determine, and choosing $\lambda$ too high will cause the algorithm to converge slowly. See \cite{puter} Chapter $10$ for more on Blackwell optimality.
In VIA we calculate
\[
v^{n+1}(x_t,\theta) = \max_{u} \left\{ C(x_t,u_t, \theta) +\lambda\cdot E_{x_{t+1}}[ v^n(x_{t+1}, \theta)|x_t,u_t,\theta ]   \right \}
\]
in a while-loop until $v^n$ converges to some fixed point, within some tolerance. Convergence is guaranteed since each iteration of $v^n$ is a contraction mapping.

%

Our aim with VIA is to maximize  the average Partial Observation Fisher Information
\[
\lim_{n \rightarrow \infty}\frac{1}{n} E_{\theta}E_{{\bf y}|\theta}\sum_{t=0}^{n}\left( \frac{\partial}{\partial \theta} \log p(y_{t+1}| y_{(t-m):t},u_{(t-m):t},\theta)\right)^2
\]
in order to obtain a time-invariant policy. In the FOFI case, this can be replaced by the Full Observation Fisher Information.


We propose running VIA at every time step $t$, but to use the posterior for $\theta$, $\pi(\theta|y_{0:t}, u_{0:(t-1)})$, which is conditioned on all the data observed so far, instead of using the prior for $\theta$. This will give a control that maximizes the average Fisher Information, using all the parameter information that is available at time $t$.
Instead of starting VIA at each time $t$ with $v^1=0$, considerable time can be saved by using the last value vector $v^n$ from the previous run of VIA at time $t-1$. This is because the posterior for $\theta$ often doesn't change much between time steps, and the last $v^n$ from time $t-1$ thus being relatively close to the fixed point at time $t$.

 The pseudocode for this modified VIA using POFI is provided in Appendix \ref{sec:pseudovia}.
Updating FOFI policies online using VIA can be done in a similar way. In the next example we compare FOFI and POFI both employing an expectation over the prior and using the VIA updates.

\subsection{PCR Model} \label{subsec:pcr}

Polymerase chain reaction is a well established method to copy and multiply DNA. We are interested in modeling the growth dynamics of DNA template ($x_t$), for a fixed amount of substrate. The model we use is
\[
x_{t+1}=(1-u_t)x_t +dt \frac{a(1-u_t) x_t}{(b+(1-u_t)x_t)^2} +\sqrt{dt}\cdot \varepsilon_1
\]
 where $\varepsilon_1 \sim N(0,\sigma_1^2)$. Here $x_t$ is the amount of DNA template, $a$ and $b$ the parameters of the model and $u_t$ the control, the percentage of template removed at each time point. We are interested in estimating the parameter $b$, labeled the half-saturation constant. A good reference for PCR models is \cite{haccou2005branching}.

We measure the amount of DNA template at each time point, but with an error. Our observations are
$
y_t = x_t +\varepsilon_2$ where  $\varepsilon_2 \sim N(0,\sigma_2^2)
$
and thus we have a dynamical system which when discretized becomes a Partially Observed Markov Decision Process.

The range for $x_t$ was set to be $[0,15]$ and then discretized into $200$ intervals, and $y_t$ was discretized to the same range, but only into $50$ intervals. The parameter values were set to be $a= 2.0$, $b=4.2$, $\sigma_1=\sigma_2=1$, $dt=1$ and the possible values of the control $u_t \in \{0,.2,.4,.6,.8,1\}$.

Still with the objective of maximizing Fisher Information, we more realistically assumed priors for the parameters of the system, as discussed above.
We conducted a simulation study using controls based on these priors for both FOFI and POFI with $m=1$, and then compared their performance to controls that are updated online using VIA, also both for POFI and FOFI. As a baseline comparison we also ran simulations using fixed controls and simulations where the true parameter is used (unrealistically) to calculate the control policy via dynamic programming as in the previous examples. For fixed controls we report the simulation with the lowest MSE, which was when $u_t = .2$.

The range for $b$ was set to be $b \in [1.7 , 8.0]$ and then we discretized that interval into 10 points $\{1.7,2.4,3.1,3.8,4.5,5.2,5.9,6.6,7.3,8.0\}$. We then considered a uniform prior on these points with a prior that  puts the weight $.9$ on the point $7.3$ and gives the others equal weight. This second ``inaccurate'' prior is intended to demonstrate the benefits of updating our knowledge of $b$ as the experiment progresses.  The discounting factor for VIA was set to be $\lambda=.9$.

Table~\ref{chemostat} reports the result of a simulation study that compares
\begin{itemize}
\item A fixed control of $u_t = 0.2$ for all times $t$,
\item Control policies calculated for the POFI and FOFI objectives averaged over uniform prior for $b$ obtained prior to the experiment,
\item Control policies calculated for the POFI and FOFI objectives averaged over the ``inaccurate'' prior,
\item  POFI and FOFI strategies averaged over a posterior that is updated using  VIA  as the experiment progresses, starting from a uniform prior,
\item POFI and FOFI strategies updated using the VIA starting from the inaccurate prior and
\item control policies obtained using the true parameter values for $b$.
\end{itemize}
Our simulation study had the time length $T=200$ and there were $600$ simulations for each case. The parameter $b$ was estimated using an EM algorithm. 

\begin{table}
 \begin{center}
  \begin{tabular}{c c}
 \begin{tabular} {|c| c c c|}
 \hline
\multicolumn{4}{|c|}{uniform prior, without VIA} \\
  \hline
\multicolumn{4}{|c|}{fixed control ($u_t=.2$ for all $t$) } \\
 \hline
  & bias & st. dev. & RMSE \\ 
 fixed &  .1264 & .7466 & 0.7572 \\ 
 \hline
 \hline
 & bias & st. dev. & RMSE \\ 
 FOFI & 0.1059 & 0.6598 & 0.6682 \\ 
POFI & 0.0053& 0.6189 & 0.617 \\ 
  \hline
   \hline
\multicolumn{4}{|c|}{inaccurate prior, without VIA}\\
  \hline
 & bias & st. dev. & RMSE \\ 
FOFI& 0.0755& 0.6374 & 0.6419 \\ 
POFI & 0.0516& 0.7051 & 0.7063 \\ 
\hline
\end{tabular}
&
 \begin{tabular} {|c| c c c|}
  \hline
\multicolumn{4}{|c|}{uniform prior, with VIA}\\
 \hline
 & bias & st. dev. & RMSE \\ 
FOFI & 0.0388 & 0.6180 & 0.6192 \\ 
POFI & 0.0766 & 0.5999 & 0.6048 \\ 
  \hline
   \hline
\multicolumn{4}{|c|}{inaccurate prior, with VIA}\\
  \hline
 & bias & st. dev. & RMSE \\ 

FOFI& 0.0713 & 0.6787 & 0.6824 \\ 
POFI & 0.0954 & 0.6750& 0.6818 \\ 
  \hline
  \hline
\multicolumn{4}{|c|}{True parameter, without VIA}\\
  \hline
 & bias & st. dev. & RMSE \\ 
FOFI & 0.0659& 0.6235 & 0.6271 \\ 
POFI & 0.0323& 0.6249 & 0.6258 \\ 

   \hline
 \end{tabular}
 \\
\end{tabular}
\caption{Simulation results for the PCR Model using two kinds of priors, POFI and FOFI and with and without VIA.}
 \label{chemostat}
 \end{center}
\end{table}

We note that when we calculate the controls prior to the experiment (No online updating), both the POFI and FOFI controls are significantly better than using a fixed control, and POFI seems to do better than FOFI when we use an uniform prior. Interestingly in the FOFI case, calculating the controls using the inaccurate prior does better then using the uniform prior, likely due to a reduction in prior variance, in spite of additional bias.

Accuracy increases in most cases when we allow for online updating using the VIA algorithm. Starting the VIA with an uniform prior does better than starting with the inaccurate one, which is probably due to the VIA having to spend more time ``repairing" the prior. Also, we note that VIA controls with uniform prior have a similar performance to a control policy using the true (unknown) parameter.

Additionally, in Figure~\ref{VIA_time} in Appendix \ref{sec:viaperform}, we see that using the previous final value vector as the starting value vector of VIA when going from time point $t$ to $t+1$, does save considerable time, and more so as $t$ grows and the posterior for the parameter starts to change less.

%

\section{Discussion}

In this paper we compared two possible ways to conduct experimental design in parametric POMDP's, based on using dynamic programming to maximize either the Partial Observation Fisher Information or the Full Observation Fisher Information.
Settings can arise where controls chosen by FOFI are not optimal, due to focusing on the underlying process rather than the observed process, and in these cases controls chosen with POFI often perform better, as in the six state example and the adversarial game; in other examples analyzed they performed similarly.

In recent years, there has been growing interest in statistical procedures within dynamical systems, such as parameter estimation and hypothesis testing, and many of these procedures could be performed more efficiently given good experimental design. In this paper we fully discretized the state and observational spaces to transform dynamical systems with stochastic errors into partially observed Markov decision processes, allowing us to use the methods developed for POMDP's to our advantage.

We also noted how the problem of parameter dependence can be overcome by averaging over a prior. Additionally given that there is enough time between consecutive time steps, we showed how the controls can be efficiently updated online using observations gathered so far, by using a variant of the Value Iteration Algorithm. This was demonstrated in the PCR example.

There remain many open challenges in experimental design for nonlinear processes. The methods we present are based on discretizing continuous state and observation quantites, which limits the dimension of the state variables. Extending our methods to higher-dimensional systems, or to incorporate more than lags for larger numbers of states could be approached using the techniques of approximate dynamic programming, see \cite{powell}. We have focussed on designing control variables, but other design quantities such as the timing or type of observations can be important. Finally, we have focussed only one particular design objective within the framework of Fisher Information. Criteria such as the trace of the Fisher Information can be targeted in our framework when incorporating multiple parameters. Mutual Information was explored in \cite{iolov2017} for a particular system (also with one parameter of interest).  Other targets such as the power of a test or in model selection (for example, those in \cite{hooker2015}) have yet to be investigated.



\subsection*{Acknowledgement}
This work was partially supported by NSF grants DMS-1053252 and DEB-1353039.

\bibliographystyle{chicago}
\bibliography{mybib}

\appendix

\section{Computing and Pseudocode}  \label{app:A}
\subsection{Dynamic program for approximated POFI}
\label{sec:dyn POFI}
We maximize
\[
POFI_m = E\sum_{t=0}^{T-1}\left(\frac{\partial}{\partial \theta} \log p(y_{t+1}|y_{(t-m):t},u_{(t-m):t},\nu_{t-m},\theta)\right)^2
\]
with reward function 
$C(y_{(t-m):t},u_{(t-m):t} )= \left(\frac{\partial}{\partial\theta} \log p(y_{t+1}|y_{(t-m):t},u_{(t-m):t} ,\theta)\right)^2$ and Partial Observation Fisher Information To Go
\[
POFI_{t,m} (y_{(t-m):t},u_{(t-m):(t-1)}) = \max_{u_t} \left\{ C(y_{(t-m):t},u_{(t-m):t})+E_{y_{t+1}}[POFI_{t+1,m}(y_{(t-m):t},u_{(t-m):t} ,\theta) ]  \right \}
\]
 The pseudocode for this dynamic program is:\\

\begin{algorithmic}
\STATE $POFI_{T,m}=  0$
\FOR{$t=(T-1) \to 0$}
\STATE $\forall$ $y_{(t-m):t},u_{(t-m):(t-1)}$ and calculate and store
\STATE $POFI_{t,m}(y_{(t-m):t},u_{(t-m):(t-1)},\theta) = \max_{u_t} \left\{ E_{y_{t+1}}[C + POFI_{t+1,m}|y_{(t-m):t},u_{(t-m):t},\theta ]   \right \}$
\STATE $u_t^*(y_{(t-m):t},u_{(t-m):(t-1)},\theta)=\argmax_{u_t}  \left\{ E_{y_{t+1}}[C +POFI_{t+1,m}|y_{(t-m):t},u_{(t-m):t},\theta ]   \right \}$
\ENDFOR
\end{algorithmic}

\subsection{Dynamic program for FOFI}
\label{sec:dyn FOFI}
We maximize
\[
FOFI=E\sum_{t=0}^{T-1}\left(\frac{\partial}{\partial \theta} \log p(x_{t+1}|x_t,u_t,\theta)\right)^2
\]
by setting the reward function as $C(x_t,u_t)=\left(\frac{\partial}{\partial \theta} \log p(x_{t+1}|x_t,u_t,\theta)\right)^2$.
The pseudocode for this dynamic program is:\\

\begin{algorithmic}
\STATE $FOFI_T=  0$
\FOR{$t=(T-1) \to 0$}
\STATE $\forall$ $x_t$ and calculate and store
\STATE $FOFI_t(x_t) = \max_{u_t} \left\{ C(x_t,u_t,\theta) + E_{x_{t+1}}[FOFI_{t+1}(x_{t+1}, \theta)|x_t,u_t,\theta ]   \right \}$
\STATE $u_t^*(x_t)=\argmax_{u_t}  \left\{C(x_t,u_t,\theta) +  E_{x_{t+1}}[FOFI_{t+1}(x_{t+1},\theta)|x_t,u_t,\theta ]   \right \}$
\ENDFOR
\end{algorithmic}

\subsection{Value Iteration Algorithm Pseudocode} \label{sec:pseudovia}

The code below provides a formal algorithm for the modified VIA employed in Section \ref{sec:via}.

Let $v_t^n$ denote the value vector at time $t$ at the $n$'th iteration of the $t$'th VIA and let \mbox{$\pi(\theta|y_{0:t}, u_{0:(t-1)})$} denote the posterior for $\theta$ given observations up till time $t$. Also, to ease notation, let ${\bf z_t} = y_{(t-m):t},u_{(t-m):(t-1)}$. Then

 \vspace{4 mm}
\begin{algorithmic}
\STATE Set $v_1^0 = 0$ and $n=0$
\FOR{$t=0 \to T$}
\WHILE{$||v^n-v^{n-1}|| >\varepsilon$}
\STATE $\forall$ ${\bf z_t}$ and calculate and store
\STATE $  v_t^{n+1}({\bf z_t})=$
\[ \max_{u_t}  \sum_{\theta} \sum_{ y_{t+1}}\left[\left( \frac{\partial}{\partial \theta} \log p(y_{t+1}|{\bf z_t},u_t,\theta)\right)^2 + \lambda v_t^n({\bf z_{t+1}}))p(y_{t+1}|{\bf z_t},u_t,\theta)\pi(\theta|y_{0:t}, u_{0:(t-1)})\right]
\]
\STATE n=n+1
\ENDWHILE
\STATE Set $v_{t+1}^0=v_t^n$
\STATE Now let
\STATE $ u_t({\bf z_t})=$
\[ \argmax_{u_t}  \sum_{\theta} \sum_{y_{t+1}}\left[\left( \frac{\partial}{\partial \theta} \log p(y_{t+1}|{\bf z_t},u_t,\theta)\right)^2 + \lambda v_t^n({\bf z_{t+1}}))p(y_{t+1}|{\bf z_t},u_t,\theta)\pi(\theta|y_{0:t}, u_{0:(t-1)})\right]
\]
\STATE Use control $u_t$, and observe $y_{t+1}$ and then update the posterior for $\theta$,
\[
\pi(\theta|y_{0:(t+1)}, u_{0:t})=\frac{p(y_{t+1}|y_{0:t},u_{0:t},\theta)\pi(\theta|y_{0:t}, u_{0:(t-1)})}{\sum_{\theta} p(y_{t+1}|y_{0:t},u_{0:t},\theta)\pi(\theta|y_{0:t}, u_{0:(t-1)})}
\]
\ENDFOR
\end{algorithmic}
\vspace{4 mm}

\subsection{Computational Performance} \label{app:B1}

We analyze the computational complexities of FOFI and POFI. The computations required can be split into computations done prior to the experiment, and computations that are required while running the experiment. A direct comparison is not completely fair since FOFI requires computations at runtime while POFI does not as discussed below.
\subsubsection{FOFI}
Prior to the experiment we use a dynamic program that provides us with a control-policy that maximizes
\[
E\sum_{t=0}^{T-1}\left(\frac{\partial}{\partial \theta} \log p(x_{t+1}|x_t,u_t,\theta)\right)^2,
\]
i.e. the Full Observation Fisher Information.

We assume that the transition probability matrix $p(x_{t+1}|x_t, u_t,\theta)$ is given. Calculating $ \left(\frac{\partial}{\partial \theta} \log p(x_{t+1}|x_t, u_t,\theta)\right)^2$ is negligible compared to the calculations required for the dynamic program. If we set
\[
g_t(x_t,x_{t+1},u_t,\theta) =\left(\frac{\partial}{\partial \theta} \log p(x_{t+1}|x_t, u_t,\theta)\right)^2
\]
then for a given time $t$ in the dynamic program we need to maximize
\[
 E\left[g_t(x_t,x_{t+1},u_t,\theta) +V_{t+1}(x_{t+1},\theta)\middle| x_t\right]
\]
over $u_t \in \mathcal{U}$ for each $x_t \in \mathcal{X}$, where $V_{t+1}$ is the value function from the previous step $t+1$.
This calculation requires adding $g_t$ and $V_{t+1}$ which are two $K^{\times 2}\times l$ tensors with cost $K^2l$. Next we need a dot product between $g_t+V_{t+1}$ and $p(x_{t+1}|x_t,u_t)$ over the $x_{t+1}$ dimension which has cost $O(K^2l)$. Finally maximizing over $u_t$ for each $x_t$  has cost $O(Kl)$. Thus each step $t$ has cost $O(K^2 l)$ and the dynamic program in total has cost $O(TK^2l)$.

When running the experiment, a filter is required to estimate the state $x_t$. The filter for time $t+1$ can be calculated via the following recursive formula:
\[
 p(x_{t+1}|y_{0:(t+1)},u_{0:t}) \propto \sum_{x_t} p(y_{t+1}|x_{t+1})p(x_{t+1}|x_t,u_t)p(x_t|y_{0:t},u_{0:(t-1)})
\]
and then normalizing. This requires $2K$ dot products of vectors of length $K$, with cost $O(K^2)$ and the normalization has cost $O(K)$. Thus we have $O(K^2)$ computations at each time step $t$ during runtime.

\subsubsection{POFI}
Here the dynamic program maximizes the approximated Partial observation Fisher Information,
\[
E\left[\sum_{t=0}^{T-1}\left(\frac{\partial}{\partial\theta} \log p(y_{t+1}| y_{(t-m):t},u_{(t-m):t},\theta) \right)^2  \right ]
\]
First we note that $ p(y_{t+1}|y_{(t-m):t},u_{(t-m):t},\theta)$ is a $L^{\times m+2} \times l^{\times m+1}$ tensor, and it can be calculated using Bayes rule at the cost $O(K^2 L^{m+2} l^{m+1})$. Calculating $\frac{\partial}{\partial \theta} \log p(y_{t+1}|y_{(t-m):t},u_{(t-m):t})$ can also be done at the cost $O(K^2 L^{m+2} l^{m+1})$, but can also be effectively approximated using the finite difference approximation to the derivative.

The cost analysis of the POFI dynamic program is just like the analysis of FOFI. At a given time $t$ adding $g_t$ and $V_{t+1}$ has cost $O(L^{m+2}l^{m+1})$, the dot product between $g_t + V_{t+1}$ and $ p(y_{t+1}|y_{(t-m):t},u_{(t-m):t})$ has cost $O(L^{m+2}l^{m+1})$ and the maximization has cost $O(L^{m+1}l^{m+1})$.

The dynamic program thus has cost $O(TL^{m+2}l^{m+1})$, which we kept from growing to large by choosing $L$ significantly lower than $K$ and $m=1$ or $m=2$.

\section{Further Experimental Results}

\subsection{Further Results on Discrete Systems}

Table \ref{6 state POFI lag 3} presents a tabulation of the POFI ($m=2$) control policy for the 6-state example.  Table \ref{adversarial results2} presents extended results for the Gamble-Safe game example including POFI with $m$ up to 7.

\begin{table}
\begin{center}
\begin{tabular} {|c || c c c c c c c c c c  c c c c c c |}
\hline
$u_t$  & 1&-1 &1 &-1& -1&-1&-1&1& 1& -1& 1&-1&-1 &-1&-1&1\\

\hline
\hline
$z_t$       &1 & 2  &  1 &  2  & 1 & 2  & 1  &  2   &1 & 2  &  1 &  2  & 1 & 2  & 1  &  2\\
$u_{t-1}$ &1 &1  & -1 &  -1  & 1 &1 & -1 & -1  &1 & 1  &  -1 &  -1  & 1 & 1  & -1  &  -1  \\
$z_{t-1}$ &1 &1  & 1 &  1  & 2 & 2  &  2  &  2  &1 & 1  &  1 &  1  & 2 & 2  & 2  &  2 \\
$u_{t-2}$ &1 &1  & 1  &  1  & 1 & 1  &  1  &  1  &-1 &-1  & -1  &  -1  & -1 & -1  &  -1  &  -1  \\
$z_{t-2}$ &1 &1  & 1  &  1  & 1 & 1  &  1  &  1  &1 & 1  &  1 &  1  & 1 & 1  & 1  &  1 \\

\hline

\hline

$u_t$  &-1 &-1 &1&-1&-1 &1&-1&1&-1&-1 &1 &-1& -1&1&-1&1\\
\hline
\hline

$z_t$        &1 & 2  &  1 &  2  & 1 & 2  & 1  &  2   &1 & 2  &  1 &  2  & 1 & 2  & 1  &  2\\
$u_{t-1}$ &1 &1  & -1 &  -1  & 1 &1 & -1 & -1  &1 & 1  &  -1 &  -1  & 1 & 1  & -1  &  -1  \\
$z_{t-1}$ &1 &1  & 1 &  1  & 2 & 2  &  2  &  2  &1 & 1  &  1 &  1  & 2 & 2  & 2  &  2 \\
$u_{t-2}$ &1 &1  & 1  &  1  & 1 & 1  &  1  &  1  &-1 &-1  & -1  &  -1  & -1 & -1  &  -1  &  -1  \\
$z_{t-2}$ &2 & 2  & 2  &  2 & 2 & 2  &  2  &  2  & 2 & 2  &  2 &  2  & 2 & 2  & 2  &  2 \\

\hline
\end{tabular}

\end{center}
\caption{Long run control policy that results from using POFI with $m=2$ in the 6 state example. The first row describes which control to use for a given history $(z_t, u_{t-1}, z_{t-1}, u_{t-2}, z_{t-2})$ of observations and control. We see that if $z_t = z_{t-1}$ and ($z_{t-1} = z_{t-2}$ or $u_{t-1} = -1$) then $u_t = 1$, else $u_t = -1$.
}
\label{6 state POFI lag 3}
\end{table}

\begin{table}
\begin{center}
\begin{tabular} {|c| c c c c|}
\hline
 & m & bias & st. dev. & RMSE \\ 
 \hline
 FOFI & - & 0.0051  & 0.2675 & 0.2676 \\ 
 POFI & 0 & 0.0115 & 0.2717 & 0.2718 \\ 
  POFI & 1 & 0.0022 & 0.2712 & 0.2711 \\ 
 POFI & 2 & 0.0096 & 0.2565 & 0.2567 \\ 
  POFI & 3 & 0.0100 & 0.2489 & 0.25 \\ 
 POFI & 4 & 0.0223 & 0.2428 & 0.2439 \\ 
 POFI & 5 & 0.0094 & 0.2448 & 0.2449 \\ 
 POFI & 6 & 0.0051 & 0.2418 & 0.2419 \\ 
 POFI & 7 & 0.0027 & 0.2419 & 0.2419 \\ 
 Random & - & 0.0075 & 0.2667 & 0.2668 \\ 
 \hline
\end{tabular}
\end{center}
\caption{Simulation results for the adversarial game. The POFI control policy, with $m>1$, offers the best controls for optimizing the estimation of $\theta$ with MSE at or under $.6$}
\label{adversarial results2}
\end{table}

\subsection{Practical Computational Performance for VIA} \label{sec:viaperform}

Figure \ref{VIA_time} plots the running time of each step of the VIA algorithm in the PCR model; as the algorithm progresses, update-time steadily reduces as parameter estimates stabilize.

\begin{figure}
\centering
\includegraphics[scale=0.31]{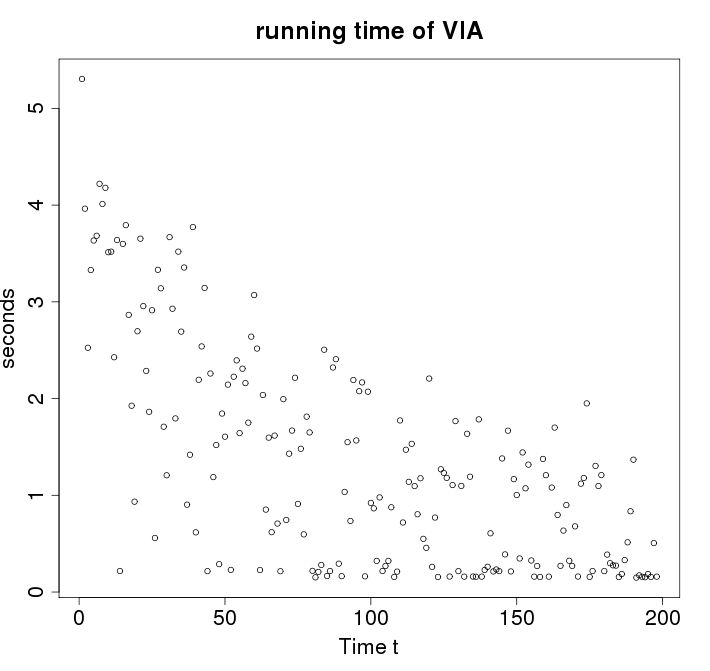}
\caption{Running time of VIA at each time step $t$, for POFI using a uniform prior for the PCR model.}
\label{VIA_time}

\end{figure}

\section{Theoretical Identities and Proofs}

\subsection{Expressing Fisher Information}
\label{sec:express FI}


In this section we find useful expressions for the Fisher Information of an experiment, that are needed to derive convergence arguments and set up dynamical programs.
We use the short hand notation
\[
 h_{k}(\theta) = \log p(y_k|y_{0:(k-1)},u_{0:(k-1)},x_0=x)
\]
where the dependence on $x_0=x$ is suppressed.

For data $y_0,\ldots, y_T$ the Fisher Information for $\theta$ can be expressed in one or two derivatives
\[
 FI =  E\left[\sum_{t=0}^{T-1} -\ddot{h}_{t+1}\right]=  E\left[\left(\sum_{t=0}^{T-1} \dot{h}_{t+1}\right)^2 \right]
\]
In order for these expressions to be well defined, we require some regularity conditions of the model. Within the POMPD framework we employ here, the expectation is taken over a finite state and observation space. Applying chain rule to $\log h_k(\theta)$ these hold provided the following conditions hold:
\begin{enumerate}
\item $p(y_t=y^i|x_{t-1}=x^j,u_{t-1} = u^r; \theta)>0$ for all $y^i \in \mathcal{Y}$, $x^j \in \mathcal{X}$ and $u^r \in \mathcal{U}$ at $\theta$.

\item $p(y_t=y^i|x_t=x^j; \theta)$ and $p(x_{t+1} = x^k|x_{t}=x^l,u_t=u^r,\theta)$ are both twice continuously differentiable at $\theta$ for each $y^i$, $x^j$, $x^k$, $x^l$ and $u^r$.
\end{enumerate}
The first of these is somewhat stronger than Assumption \ref{mixing} in which the support of $p(y_t|x_{t-1},u_t,\theta)$ can depend on $u_t$ (for this to happen, the support of $y_t$ must depend on $x_t$).

These conditions hold for the discrete implementation of all of our examples. For continuous state and observation spaces, we would need further regularity, not only to interchange integration and differentiation, but to also assume that $x_t$ cannot diverge too quickly. Rigorously pursuing such conditions is beyond the scope of this paper.

We can similarly define the Fisher Information to Go as
\[
 FI_k(y_{0:k},u_{0:(k-1)}) = E\left[\sum_{t=k}^{T-1} -\ddot{h}_{t+1}\middle|y_{0:k},u_{0:(k-1)}\right]
 = E\left[\left(\sum_{t=k}^{T-1} \dot{h}_{t+1}\right)^2  \middle|y_{0:k},u_{0:(k-1)}\right]
\]
where the latter equality is conditional on the Fisher Information being well defined as described above. We set $FI_0 = FI$.

The Fisher Information to Go can be calculated recursively (in both one or two derivatives):
\begin{mylemma}
\[
 FI_k(y_{0:k},u_{0:(k-1)}) = E\left[ -\ddot{h}_{k+1} + FI_{k+1}\middle|y_{0:k},u_{0:(k-1)}\right] = E\left[\left( \dot{h}_{k+1}\right)^2+FI_{k+1}\middle|y_{0:k},u_{0:(k-1)}\right]
\]
\label{FI recursive}
\end{mylemma}

\begin{proof}
In the case of using two derivatives this follows from iterated expectation. In one derivative we have
\begin{align*}
 FI_k(y_{0:k},u_{0:(k-1)}) & = E\left[\left(\sum_{t=k}^{T-1} \dot{h}_{t+1}\right)^2  \middle|y_{0:k},u_{0:(k-1)}\right] \\
 & = E\left[\left( \dot{h}_{k+1}\right)^2 +\left(\sum_{t=k+1}^{T-1} \dot{h}_{t+1}\right)^2
 +2\left( \dot{h}_{k+1}\right)\left(\sum_{t=k+1}^{T-1} \dot{h}_{t+1}\right)\middle|y_{0:k},u_{0:(k-1)}\right]
\end{align*}
The cross term is
\begin{align*}
& E\left[2\left( \dot{h}_{k+1}\right)\left(\sum_{t=k+1}^{T-1} \dot{h}_{t+1}\right)\middle|y_{0:k},u_{0:(k-1)}\right] \\
& \hspace{1cm} =E\left[ E\left[2\left( \dot{h}_{k+1}\right)\left(\sum_{t=k+1}^{T-1} \dot{h}_{t+1}\right)\middle| y_{0:(k+1)},u_{0:k}\right]\middle|y_{0:k},u_{0:(k-1)}\right] \\
& \hspace{1cm} =E\left[2\left( \dot{h}_{k+1}\right) E\left[\left(\sum_{t=k+1}^{T-1} \dot{h}_{t+1}\right)\middle| y_{0:(k+1)},u_{0:k}\right]\middle|y_{0:k},u_{0:(k-1)}\right] \\
& \hspace{1cm}=E\left[2\left( \dot{h}_{k+1}\right)\cdot 0\middle|y_{0:k},u_{0:(k-1)}\right] = 0
\end{align*}
Thus
\begin{align*}
 FI_k(y_{0:k},u_{0:(k-1)}) & = E\left[\left( \dot{h}_{k+1}\right)^2 +\left(\sum_{t=k+1}^{T-1} \dot{h}_{t+1}\right)^2 \middle|y_{0:k},u_{0:(k-1)}\right] \\
& = E\left[\left( \dot{h}_{k+1}\right)^2 +E\left[\left(\sum_{t=k+1}^{T-1} \dot{h}_{t+1}\right)^2\middle| y_{0:(k+1)},u_{0:k}\right] \middle|y_{0:k},u_{0:(k-1)}\right] \\
& = E\left[\left( \dot{h}_{k+1}\right)^2+FI_{k+1}\middle|y_{0:k},u_{0:(k-1)}\right].
\end{align*}
\end{proof}

\begin{mycorollary}
\[
FI =  E\left[\sum_{t=0}^{T-1} \left(\dot{h}_{t+1}\right)^2\right]
\]
and similarly
\[
FI_k(y_{0:k},u_{0:(k-1)}) =  E\left[\sum_{t=k}^{T-1} \left(\dot{h}_{t+1}\right)^2 \middle|y_{0:k},u_{0:(k-1)}\right]
\]
\label{FI express}
\end{mycorollary}
\begin{proof}
This follows from using induction and Lemma \ref{FI recursive}.
\end{proof}
%
%
%

\subsection{Approximating the Fisher Information to Go}

Running an exact dynamic program, with $FI_k(y_{0:k},u_{0:(k-1)})$ as the value function, is not computationally feasible, leading us to approximate it by, at each time $t$, examining only the previous $m$ observations.
%
%
We set

\[
 h_{k,m,\nu_m}(\theta) = \left\{
  \begin{array}{l l}
    \log p(y_k|y_{m:(k-1)},u_{m:(k-1)},\nu_m) & \quad \text{if $m \geq 0$}\\
    \log p(y_k|y_{0:(k-1)},u_{0:(k-1)},\nu_0) & \quad \text{if $m < 0$}
  \end{array} \right.
\]
where $\nu_m$ is the assumed distribution of $x_m$ and we will consider it to be fixed and known. Allowing $m$ to be negative will ease notation when $t-m<0$.
%
%
We now set
\begin{align*}
 POFI_{k,m}(y_{(k-m):k},u_{(k-m):(k-1)}) & = E\left[\sum_{t=k}^{T-1}\left(\dot{h}_{t+1,t-m,\nu_{t-m}}\right)^2\middle|y_{(k-m):k},u_{(k-m):(k-1)}\right] \\
& = E\left[\sum_{t=k}^{T-1}-\ddot{h}_{t+1,t-m,\nu_{t-m}}\middle|y_{(k-m):k},u_{(k-m):(k-1)}\right]
\end{align*}
and denote it the Partial Observation Fisher Information to Go. That the formulation in one derivative is equal to the one in two derivatives follows from the individual parts of each sum having a Fisher Information interpretation.

%
%

\subsection{Mixing conditions}
\label{sec:mixing}
\citet{Cappe} establish forgetting properties of the filter by assuming mixing conditions for Hidden Markov Models. These conditions are given in Assumption \ref{mixing}, slightly modified to allow for controls.
The discussion in \citet{Cappe} on which models satisfy these conditions applies analogously to POMDP's. Given these conditions we can prove the following, which is a modification of Lemma $12.5.3$ in \citet{Cappe};

\begin{theorem} \label{thm:mixcond}
Assume the strong mixing conditions in Assumption \ref{mixing}. Then
  \[
  (E |\dot{h}_{k,0,\nu_0}(\theta) - \dot{h}_{k,m,\nu_m}(\theta)|^2)^{1/2}
  \leq 8\sup_{x,x'\in \mathcal{X},u \in \mathcal{U}, y \in \mathcal{Y}} ||\phi_{\theta}(x,x',y,u)||\frac{\rho^{(k-m)/2-1}}{1-\rho}
 \]
 where $\phi(x,x',u,y) =\frac{\partial}{\partial \theta}\log(p(x_{t+1}=x'|x_t=x,u_t=u,\theta)p(y_{t+1}=y'|x_{t+1}=x',\theta))$ and
 $ \rho = \max_{y \in \mathcal{Y}}  1 -\frac{\varsigma^-(y)}{\varsigma^+(y)}$
\label{lemma 12.5.3}
\end{theorem}
A proof is provided in Appendix \ref{app:B}.

\subsection{Bounds on Fisher Information}
\label{sec:bounds}
In this section we show that the approximated Fisher Information approaches the true Fisher Information exponentially as one conditions on more and more observations, while using the same controls.

By Corollary \ref{FI express} the Fisher Information for the POMDP is
\[
 FI(u_{0:(T-1)})
 =E\sum_{k=0}^{T-1}\left( \dot{h}_{k+1,0,\nu_0}(\theta)\right)^2
\]
but since that is computationally intractable, we consider
\[
 POFI_m= FI_{0,m}(u_{0:(T-1)}) = E\sum_{k=0}^{T-1}\left( \dot{h}_{k+1,k-m,\nu_{k-m}}(\theta)\right)^2
\]
see definitions for $\dot{h}$ above. Here we use Fisher Information in one derivative, but as noted above it is equivalent to using the formulation in two derivatives. Also note that where $k-m<0$ we just set it to $0$ and use the initial distribution of $x_0$.

\begin{mylemma} Assume the mixing conditions in Assumption \ref{mixing} hold. Then

 \begin{align*}
 & \left(E(\dot{h}_{k+1,0,\nu_0} + \dot{h}_{k+1,k-m,\nu_{k-m}})^2\right)^{1/2} \\
 & \hspace{2cm} \leq 16\sup_{x,x'\in \mathcal{X},u \in \mathcal{U}, y \in \mathcal{Y}} |\phi_{\theta}(x,x',y,u)|\frac{\rho^{1/2}}{1-\rho} + 2\sup_{u_0} \left(E(\dot{h}_{1,0,\nu_0})^2\right)^{1/2}
 \end{align*}
 \label{bound lemma}
\end{mylemma}
\begin{proof}
Set
\[
 A(m') = \sup_{u_1,\ldots, u_{m'-1}} \left(E(\dot{h}_{m',0,\nu_0})^2\right)^{1/2}
\]
which sets an upper bound on the length of $\dot{h}_m$. Note that $A(m')$ also bounds $(E(\dot{h}_{k+1,k-m,\nu_{k-m}})^2)^{1/2}$ since $\nu_{k-m} = \nu_0$.
Now
\begin{align*}
\left(E(\dot{h}_{k+1,0,\nu_0} + \dot{h}_{k+1,k-m,\nu_{k-m}})^2\right)^{1/2}
 & \leq   \left(E(\dot{h}_{k+1,0,\nu_0} - \dot{h}_{k+1,k-m',\nu_{k-m'}})^2\right)^{1/2} \\
 & \hspace{0.5cm} +  \left(E(\dot{h}_{k+1,k-m,\nu_{k-m}}- \dot{h}_{k+1,k-m',\nu_{k-m'}})^2\right)^{1/2} \\
 & \hspace{0.5cm} + 2\left(E(\dot{h}_{k+1,k-m',\nu_{k-m'}})^2\right)^{1/2} \\
 & \leq 16 \sup|\phi_{\theta}| \frac{\rho^{(1+\min(m,m'))/2}}{1-\rho} + 2A(m'+1)
\end{align*}
using Theorem \ref{lemma 12.5.3}. 
Setting $m'=0$ gives the result, although that might not be the best bound.
\end{proof}

\begin{theorem}[Restated from Theorem \ref{FI approx} in main text]
Assume the conditions in Assumption \ref{mixing} hold. Then, for $m < T$ and any control policy, we have
\[
 |FI - POFI_{0,m}| \leq c_1(T-1-m)\rho^{m/2}
 \]
 where $c_1 = 8M(\theta)\sup_{x,x'\in \mathcal{X},u \in \mathcal{U}, y \in \mathcal{Y}} |\phi_{\theta}(x,x',y,u)|\frac{1}{\rho^{1/2}(1-\rho)}$ and
$M(\theta)$ is the bound from Lemma \ref{bound lemma}; $M(\theta) = 16\sup |\phi_{\theta}|\frac{\rho^{1/2}}{1-\rho} + 2\sup_{u_0} \left(E(\dot{h}_{1,0,\nu_0})^2\right)^{1/2}$.
\label{FI approx append}
\end{theorem}

\begin{proof}
\begin{align*}
 |FI - POFI_{0,m}| & =|E\sum_{k=0}^{T-1}\left( \dot{h}_{k+1,0,\nu_0}(\theta)\right)^2 - E\sum_{k=0}^{T-1}\left( \dot{h}_{k+1,k-m,\nu_{k-m}}(\theta)\right)^2| \\
&  = |\sum_{k=m+1}^{T-1}E(\dot{h}_{k+1,0,\nu_0}^2-\dot{h}_{k+1,k-m,\nu_{k-m}}^2)| \\
&  \leq \sum_{k=m+1}^{T-1}|E( \dot{h}_{k+1,0,\nu_0}-\dot{h}_{k+1,k-m,\nu_{k-m}})
 \cdot ( \dot{h}_{k,0,\nu_0}+\dot{h}_{k,k-m,\nu_{k-m}})|
\end{align*}
and by Cauchy Schwarz the final expression is bounded by
\[
 E = \sum_{k=m+1}^{T-1}\left(E\left| \dot{h}_{k+1,0,\nu_0}-\dot{h}_{k+1,k-m,\nu_{k-m}}\right|^2\right)^{1/2}
 \cdot \left(E\left| \dot{h}_{k,0,\nu_0}+\dot{h}_{k,k-m,\nu_{k-m}}\right|^2\right)^{1/2}.
\]
For the first parenthesis in $E$ we use Theorem \ref{lemma 12.5.3} 
to get
\[
 \left(E\left| \dot{h}_{k+1,0,\nu_0}-\dot{h}_{k+1,k-m,\nu_{k-m}}\right|^2\right)^{1/2}
 \leq 8\sup_{x,x'\in \mathcal{X},u \in \mathcal{U}, y \in \mathcal{Y}} |\phi_{\theta}(x,x',y,u)|\frac{\rho^{(m+1)/2-1}}{1-\rho}
\]
and the second one is bounded by the Lemma \ref{bound lemma}. We get

\begin{align*}
  |FI - POFI_{0,m}| & \leq 8M(\theta)\sup_{x,x'\in \mathcal{X},u \in \mathcal{U}, y \in \mathcal{Y}} |\phi_{\theta}(x,x',y,u)|\sum_{k=m+1}^{T-1}\frac{\rho^{(m+1)/2-1}}{1-\rho} \\
& = 8M(\theta)(T-1-m)\sup_{x,x'\in \mathcal{X},u \in \mathcal{U}, y \in \mathcal{Y}} |\phi_{\theta}(x,x',y,u)|\frac{\rho^{(m+1)/2-1}}{1-\rho}
\end{align*}
\end{proof}

  Exactly the same arguments can be used to show that the Partial Observation Fisher Information to Go $POFI_{k,m}(y_{(k-m):k},u_{(k-m):(k-1)})$ approaches the true Fisher Information to Go as $m$ increases.

\subsection{Best Possible Fisher Information}
\label{sec:best FI}
A related problem we are interested in is how well a control policy that looks at the last say $m$ observations does in comparison with a control policy that considers all past observations. That is, we want a bound on the best possible Fisher Information given a control policy that consider all past observations, compared with the best possible Fisher Information with a control policy that only considers the last $m$ observations.

We first establish a baseline difference between two parts of the Fisher Information.
\begin{mylemma}
 Assume the conditions in Assumption \ref{mixing} hold. Then, for any control policy and any $k$, we have
 \[
  |E(\dot{h}_{k+1,0,\nu_0}^2 - \dot{h}_{k+1,k-m,\nu_{k-m}}^2)| \leq 8 M(\theta) \sup_{x,x'\in \mathcal{X},u \in \mathcal{U}, y \in \mathcal{Y}} |\phi_{\theta}(x,x',y,u)|\frac{\rho^{(m+1)/2-1}}{1-\rho}
 \]
 where $M(\theta) = 16\sup |\phi_{\theta}|\frac{\rho^{1/2}}{1-\rho} + 2\sup_{u_0} \left(E(\dot{h}_{1,0,\nu_0})^2\right)^{1/2}$ is the bound from Lemma \ref{bound lemma}.

\label{ind diff}
\end{mylemma}
\begin{proof}
 Follows from Lemma \ref{bound lemma} and Theorem \ref{lemma 12.5.3}
 as in the proof of Theorem \ref{FI approx}.
\end{proof}
As above, we assume that the Fisher Information to Go;
\[
 FI_k(y_{0:k},u_{0:(k-1)}) = E\left[\left( \dot{h}_{k+1,0,\nu_0}\right)^2+FI_{k+1}\middle|y_{0:k},u_{0:(k-1)}\right]
\]
is approximated by
\[
 POFI_{k,m}(y_{(k-m):k},u_{(k-m):(k-1)}) = E\left[\left( \dot{h}_{k+1,k-m,\nu_{k-m}}\right)^2 + POFI_{k+1,m}\middle|y_{(k-m):k},u_{(k-m):(k-1)}\right]
\]

Given that the control policy is obtained with dynamic programming, we have that the optimal control at time $t$ is dependent on the optimal control obtained at time $t+1$. Let $u_1^*,\ldots, u_{T-1}^*$ denote the set of optimal controls obtained in this manner, i.e. $u_k^*,\ldots, u_{T-1}^*$ maximize $FI_k$.

As argued these controls computationally infeasible to calculate and thus we resort to approximate optimal controls, here denoted $u_{0,m}^*,\ldots,u_{T-1,m}^*$, where $u_{k,m}^*,\ldots,u_{T-1,m}^*$ maximize $POFI_{k,m}$.

We now restate Theorem \ref{best FI} from Section $2.3$ on the loss of using an approximate control policy instead of an exact one in Fisher Information, in an experiment of length $T$.
\begin{theorem} [Restated from Theorem \ref{best FI} in the main text] Given that the mixing conditions in Assumption \ref{mixing} hold we have
\[
 0 \leq FI(u_0^*,\ldots,u_{T-1}^*) - FI(u_{0,m}^*,\ldots, u_{T-1,m}^*) \leq c_2 T(T+1)\rho^{m/2}
\]
where $c_2= 8 M(\theta) \sup |\phi_{\theta}(x,x',y,u)|\frac{1}{\rho^{1/2}(1-\rho)}$, and $M(\theta)$ is the bound from Lemma \ref{bound lemma}.
\label{best FI append}
\end{theorem}
\begin{proof}
We analyze the difference by bounding errors in each step of the dynamic program inductively, starting at time $t=T-1$ and going backwards. Set $\varepsilon = 8 M(\theta) \sup |\phi_{\theta}(x,x',y,u)|\frac{\rho^{(m+1)/2-1}}{1-\rho}$.

We find that
\begin{align*}
 0 & \leq FI_{T-1}(u_{T-1}^*) - FI_{T-1}(u_{T-1,m}^*) \\
& \leq FI_{T-1}(u_{T-1}^*) - FI_{T-1}(u_{T-1,m}^*) +(FI_{T-1,m}(u_{T-1,m}^*) - FI_{T-1,m}(u_{T-1}^*))
\end{align*}
so far only using that $u_{T-1}^*$ maximizes $FI_{T-1}$ and $u_{T-1,m}^*$ maximizes $FI_{T-1,m}$. Now
\[
 |FI_{T-1}(u_{T-1}^*)- POFI_{T-1,m}(u_{T-1}^*)| +| FI_{T-1}(u_{T-1,m}^*) - POFI_{T-1,m}(u_{T-1,m}^*) |
 \leq 2 \varepsilon
\]
by Lemma \ref{ind diff}.

We now inductively assume
\[
 |FI_{T-s}(u_{(T-s):(T-1)}^*) - FI_{T-s}(u_{(T-s):(T-1),m}^*)| \leq s(s+1)\varepsilon
\]
where $u_{(T-s):(T-1),m}^* = u_{T-s,m}^*,\ldots, u_{T-1,m}^*$, and then get
\begin{align}
 & |FI_{T-s}( u_{(T-s):(T-1)}^*) - POFI_{T-s,m}(u_{(T-s):(T-1),m}^*)|
  \\ & \hspace{2cm} \leq
  |FI_{T-s}( u_{(T-s):(T-1)}^*) - FI_{T-s}(u_{(T-s):(T-1),m}^*)| \nonumber \\
\hspace{4cm} & \hspace{1cm} + |FI_{T-s}( u_{(T-s):(T-1),m}^*) - POFI_{T-s,m}(u_{(T-s):(T-1),m}^*)| \nonumber \\
\hspace{2cm} & \leq s(s+1)\varepsilon + s\varepsilon \nonumber \\
\hspace{2cm} & =s(s+2)\varepsilon \label{eq:tag1}
\end{align}
Now moving from $s$ to $s+1$ we have
\[
 POFI_{T-(s+1),m}(u_{(T-(s+1)):(T-1),m}^*) \geq POFI_{T-(s+1),m}(u_{T-(s+1)}^*,u_{(T-s):(T-1),m}^*)
\]
since $u_{(T-(s+1)):(T-1),m}^*$ are the controls that maximize $POFI_{T-(s+1),m}$. By adding and subtracting the same quantity we get the following equivalent inequality
\begin{align}
& (POFI_{T-(s+1),m}(u_{(T-(s+1)):(T-1),m}^*)-FI_{T-(s+1)}(u_{(T-(s+1)):(T-1),m}^*)) \label{eq:tag2} \\
 & \hspace{1cm}  -(POFI_{T-(s+1),m}(u_{T-(s+1)}^*,u_{(T-s):(T-1),m}^*)-FI_{T-(s+1)}(u_{(T-(s+1)):(T-1)}^*)) \label{eq:tag3} \\
& \hspace{2cm} \geq  FI_{T-(s+1)}(u_{(T-(s+1)):(T-1)}^*)- FI_{T-(s+1)}(u_{(T-(s+1)):(T-1),m}^*) \geq 0
\end{align}
\eqref{eq:tag2} is bounded by $(s+1)\varepsilon$ by Lemma \ref{ind diff} and \eqref{eq:tag3} by $\varepsilon + s(s+2)\varepsilon$ using \eqref{eq:tag1} and Lemma \ref{ind diff}. Therefore
\begin{align*}
& |FI_{T-(s+1)}(u_{(T-(s+1)):(T-1)}^*)- FI_{T-(s+1)}(u_{(T-(s+1)):(T-1),m}^*)| \\
& \hspace{2cm}  \leq (s+1)\varepsilon + \varepsilon +s(s+2)\varepsilon = (s+1)(s+2)\varepsilon
\end{align*}
and for the whole experiment we find
\[
 |FI(u_{0:(T-1)}^*) - FI(u_{0:(T-1),m}^*)| \leq T(T+1)\varepsilon
\]
\end{proof}

\section{Modified HMM theory} \label{app:B}
This section is devoted to expanding Hidden Markov Model Theory to Partially Observed Markov Decision Processes. We base it completely on \citet{Cappe} and use their notation, only changing what is necessary. The purpose is to prove Theorem \ref{thm:mixcond} in Appendix \ref{sec:mixing}, which is a modified version of Lemma $12.5.3$ in \citet{Cappe}. In most cases the changes will amount to adding controls and seeing that the theory follows through, although the proof of Theorem 3 has more substantial changes.
\subsection{Setup}

Let $(X,\mathcal{X})$ and $(Y,\mathcal{Y})$ be the state space and the observations space respectively. Let
\[
 Q^u(x,A) = \int_A q^u(x,x')dx', \textrm{  } A \in \mathcal{X}, u \in \mathcal{U}
\]
be a transition kernel for our state space, where $u$ is a control, and $\mathcal{U}$ is finite. Also let
\[
 G(x,A) = \int_A g(x,y)dy, \textrm{   } A \in \mathcal{Y}
\]
be the transition kernel for moving from the state space to the observation space.
%

We generally assume that the Markov Chain is initialized with distribution $\nu$, and then runs for $n$ steps $x_{0:n}=x_0,\ldots,x_n$ and that $n-1$ decisions are made on what controls $u$ to use. This results in $n$ observations $y_{0:n} = y_0,\ldots, y_n$ and $n-1$ control $u_{0:(n-1)} = u_0,\ldots,u_{n-1}$.

\subsection{Hidden Markov Model theory}
\begin{mydef}[Definition 3.1.6 in \citet{Cappe}]
 Conditional on $y_{0:k}$ and $u_{0:(k-1)}$ we define the forward variable
 \[
  \alpha_{\nu,k}(y_{0:k},u_{0:(k-1)},f) =\int \cdots \int f(x_k) \nu(dx_0)g(x_0,y_0)\prod_{l=1}^k Q^{u_{l-1}}(x_{l-1},dx_l)g(x_l,y_l)
 \]
and conditional on $y_{(k+1):n}$ and $u_{k:n}$ we define the backward variable
\[
 \beta_{k|n}(y_{(k+1):n},u_{k:n},x) = \int \cdots \int Q^{u_k}(x,dx_{k+1}) g(x_{k+1},y_{k+1})\prod_{l=k+2}^{n} Q^{u_{l-1}}(x_{l-1},dx_l)g(x_l,y_l).
\]

\end{mydef}
As in the classical case, these satisfy recursion formulas
\[
 \alpha_{\nu,k}(y_{0:k},u_{0:(k-1)},f) = \int f(x_k) \int \alpha_{\nu,k-1}(y_{0:(k-1)},u_{0:(k-2)},dx_{k-1})Q^{u_{k-1}}(x_{k-1},dx_k)g(x_k,y_k)
\]
with initial condition
\[
 \alpha_{\nu,0}(f) = \int f(x_0)g(x_0,y_0)\nu(dx_0)
\]
and similarly
\[
 \beta_{k|n}(y_{(k+1):n},u_{k:n},x) =\int Q^{u_k}(x,dx_{k+1})g(x_{k+1},y_{k+1}) \beta_{k+1|n}(y_{(k+2):n},u_{(k+1):n},x_{k+1}).
\]

A standard result in HMM theory is that, conditional on the observations $y_{0:n}$, the Process $\{X_k\}_{k\geq 0}$ still is a Markov Chain, although non-homogeneous, with a transition kernel called the Forward Smoothing Kernel. We state the transition kernel here for our case, also conditional on the controls.

\begin{mydef}[Definition 3.3.1 in \citet{Cappe}]{Forward Smoothing Kernels.}
 Given $n\geq 0$ define  the transition kernels for indices $k\in \{0,\ldots,n-1\}$:
 \[
  F_{k|n}(x,A,y_{(k+1):n},u_{k:n}) = \frac{\int_A Q^{u_k}(x,dx_{k+1})g(x_{k+1},y_{k+1}) \beta_{k+1|n}(x_{k+1})}{\beta_{k|n}(x)}
 \]
\end{mydef}
Note that the Forward Smoothing Kernels are defined in terms of the backward variables.

We are generally interested in calculating smoothers and filters for our POMDP.
\begin{mydef}[Definition 3.1.3 in \citet{Cappe}]
 We let $\phi_{\nu,k:l|n}$ denote the conditional distribution of $X_{k:l}$ given $Y_{0:n}$ and in our case $u_{0:(n-1)}$ as well.
\end{mydef}

The Forward Smoothing Kernel allows us a convenient way of calculating the smoothing distributions. We first compute all the backward variables $\beta_{k|n}$ using the backward recursion given. We then note that $\phi_{\nu,0|n}$ can be calculated as
\[
 \phi_{\nu,0|n}(A) = \frac{\int_A \nu(dx_0)g(x_0,y_0)\beta_{0|n}(x_0)}{\int \nu(dx_0)g(x_0,y_0)\beta_{0|n}(x_0)}
\]
and then we have the following recursion
\[
 \phi_{\nu,k+1|n}(x) = \int \phi_{\nu,k|n}(dx_k)F_{k|n}(x_k,x) = \phi_{\nu,k|n}F_{k|n}
\]
where $F_{k|n}$ are the forward kernels, and the last equation is a short hand way of writing the integral.

Using this recursion repeatedly allows to express the smoother in the following way
\[
 \phi_{\nu,k|n}[y_{0:n}, u_{0:(n-1)}] = \phi_{\nu,0|n} \prod_{i=1}^k F_{i-1|n} [y_{i:n},u_{(i-1):(n-1)}].
\]
\subsection{Total Variation and the Dobrushin Coefficient}
To continue towards forgetting properties we need to introduce Total Variation (see Definition 4.3.1 in \citet{Cappe}). Let $\xi$ be a signed measure (it can be negative) and let $\xi = \xi_+ -\xi_-$ where $\xi_+,\xi_-$ are (positive) measures. So if $X$ is the state space
\[
 ||\xi||_{TV} = \xi_+(X) +\xi_-(X)
\]
To define the Dobrushin Coefficient (see Definition 4.3.7 in \citet{Cappe}). Let $K$ be a transition Kernel from $X$ to $Y$, the Dobrushin coefficient is given by
\[
 \delta(K) = \frac{1}{2} \sup_{(x,x')\in X\times X} ||K(x,\cdot) - K(x',\cdot)||_{TV}
\]
The Dobrushin coefficient is sub-multiplicative (see Prop. 4.3.10 in \citet{Cappe}). If $K:X\rightarrow Y, R: Y\rightarrow Z$ are $2$ transition kernels we have
\[
 \delta(KR) = \delta\left( \int K(\cdot,dx)R(x,\cdot)\right) \leq \delta(K)\delta(R)
\]

It can be shown that $ 0 \leq \delta(K)\leq 1$, however to get forgetting properties we often need $\delta(K) \leq 1-\varepsilon$, where $\varepsilon >0$.

The latter inequality holds if we assume the Doeblin Condition is satisfied:

\begin{myassump}[Assumption 4.3.12 in \citet{Cappe}]
 There exist an integer $m \geq 1, \epsilon \in (0,1)$, and a probability measure $\nu$ on $(X,\mathcal{X})$ such that for any $x\in X$ and $A \in \mathcal{X}$,
 \[
  Q^m(x,A) \geq \varepsilon \nu(A)
 \]
\end{myassump}
Under these assumptions Lemma 4.3.13 in \citet{Cappe} gives $\delta(Q^m)\leq 1-\varepsilon$.

When considering forgetting properties it is reasonable to expect that the filter $\phi_{\nu,k|n}$ depends less and less on the initial distribution of $X_0\sim\nu$, as $k$ increases. Specifically we have that when comparing initial distributions $\nu$ and $\nu'$:
\begin{align*}
& \phi_{\nu,k|n}(y_{0:n},u_{0:(n-1)},x_k) -\phi_{\nu',k|n}(y_{0:n},u_{0:(n-1)},x_k) \\
& = \int \cdots \int \left( \phi_{\nu,0|n}(y_{0:n},u_{0:(n-1)},x_k) -\phi_{\nu',0|n}(y_{0:n},u_{0:(n-1)},x_k) \right)\prod_{i=1}^k F_{i-1|n} (x_{k-1},x_k)
\end{align*}

Now using Corollary 4.3.9 in \citet{Cappe} we have
\[
 ||\xi K - \xi'K||_{TV} \leq \delta(K) || \xi - \xi'||_{TV}
\]
where $\xi,\xi'$ are probability measures, $K$ a transition kernel.

Using this on our representation of the filters gives
\[
 ||\phi_{\nu,k|n} - \phi_{\nu',k|n}||_{TV} \leq \delta\left(\prod_{i=1}^k F_{i-1|n}(y_{i:n},\cdot)\right)||\phi_{\nu,0|n} - \phi_{\nu',0|n}||_{TV}
\]
and since the Dobrushin coefficient is sub-multiplicative
\[
\leq \prod_{i=1}^k\delta\left( F_{i-1|n}(y_{i:n},\cdot)\right)||\phi_{\nu,0|n} - \phi_{\nu',0|n}||_{TV}
\]
and since the Dobrushin coefficient $\delta$ satisfies $0\leq \delta \leq 1$ we at least have that the difference between the two filters is non-expanding.

Establishing forgetting properties thus amounts to showing $\delta(F_{i-1|n}(y_{i:n})) \leq 1-\varepsilon$ for the forward smoothing kernels $F_{i|n}$. Note that so far no assumptions have been made on how quickly the Hidden Markov Model mixes. Those assumptions are made to get $\delta(F_{i|n})\leq 1-\varepsilon$.

\citet{Cappe} establish contracting bounds on the Dobrushin coefficient by imposing Strong Mixing conditions on the transition probabilities of the Hidden Markov Model.
\begin{myassump}[Assumption 4.3.21 in \citet{Cappe}]
 Strong Mixing Conditions. There exist a transition kernel $K: Y\rightarrow X$ and measurable functions $\varsigma^-$ and $\varsigma^+$ from $Y$ to $(0,\infty)$ such that for any $A \in \mathcal{X}$ and $y \in Y$,
 \[
  \varsigma^-(y)K(y,A) \leq \int_A Q(x,dx')g(x',y) \leq \varsigma^+(y) K(y,A)
 \]

\end{myassump}

In our case we have different transition kernels for each control. The weakest assumptions we can get away with is, if each transition kernel $Q^u$ has a corresponding transition kernel $K^u$ and measurable functions $\varsigma^-(y,u)$ and $\varsigma^+(y,u)$ satisfying the strong mixing condition. By letting $\varsigma^-(y) = \min_u \varsigma^-(y,u)$ and $\varsigma^+(y) = \max_u \varsigma^+(y,u)$ we see that we can consider the same $\varsigma$ functions for each transition kernel $Q^u$. We restate the Strong mixing conditions:

\begin{myassump}
 Modified Strong Mixing Conditions. For each control $u$ there exist a transition kernel $K^u: Y\rightarrow X$ and measurable functions $\varsigma^-$ and $\varsigma^+$ from $Y$ to $(0,\infty)$ such that for any $A \in \mathcal{X}$ and $y \in Y$,
 \[
  \varsigma^-(y)K^u(y,A) \leq \int_A Q^u(x,dx')g(x',y) \leq \varsigma^+(y) K^u(y,A)
 \]

\end{myassump}

Lemma $4.3.22$ in  \citet{Cappe} uses the mixing conditions stated above to establish contracting bounds on the Dobrushin coefficient. We restate the  Lemma for the POMDP case, where we also condition on the controls, and use the modified mixing conditions.

\begin{theorem}[Lemma 4.3.22 in \citet{Cappe}]
\label{lemma 4.3.22}
Under the strong mixing conditions the following holds
 \begin{enumerate} 
  \item For any non-negative integers $k$ and $n$ such that $k<n$ and $x\in X$,
  \[
   \prod_{j=k+1}^n \varsigma^-(y_j) \leq \beta_{k|n} [y_{(k+1):n},u_{k:n}] (x) \leq \prod_{j=k+1}^n \varsigma^+(y_j).
  \]
\item For any non-negative integers $k$ and $n$ such that $k<n$ and any probability measures $\nu$ and $\nu'$ on $(X,\mathcal{X})$,
\[
 \frac{\varsigma^-(y_{k+1})}{\varsigma^+(y_{k+1})}\leq \frac{\int \nu(dx) \beta_{k|n}[y_{(k+1):n},u_{k:n}](x)}{\int \nu'(dx) \beta_{k|n}[y_{(k+1):n},u_{k:n}](x)} \leq  \frac{\varsigma^+(y_{k+1})}{\varsigma^-(y_{k+1})}.
\]
\item For any non-negative integers $k$ and $n$ such that $k<n$, there exists a transition kernel $\lambda_{k|n}$ from $(Y^{n-k},\mathcal{Y}^{(n-k)})$ to $(X,\mathcal{X})$ such that for any $x \in X$, $A\in \mathcal{X}$, and $y_{(k+1):n} \in Y^{n-k}$,
\[
 \frac{\varsigma^-(y_{k+1})}{\varsigma^+(y_{k+1})}\lambda_{k,n}(y_{(k+1):n},u_{k:n},A) \leq F_{k|n}[y_{(k+1):n},u_{k:n}](x,A)
\]
\[
 \leq \frac{\varsigma^+(y_{k+1})}{\varsigma^-(y_{k+1})}\lambda_{k,n}(y_{(k+1):n},u_{k:n},A).
\]
\item For any non-negative integers $k$ and $n$, the Dobrushin coefficient of the forward smoothing kernel $F_{k|n}[y_{(k+1):n},u_{k:n}]$ satisfies
\[
 \delta(F_{k|n}[y_{(k+1):n},u_{k:n}]) \leq \rho_0(y_{k+1}) := 1 -\frac{\varsigma^-(y_{k+1})}{\varsigma^+(y_{k+1})}
\]
if $k<n$, and
\[
  \delta(F_{k|n}[y_{(k+1):n},u_{k:n}]) \leq 1 -\int \varsigma^-(y) dy
\]
if $k\geq n$.
 \end{enumerate}
\end{theorem}

\begin{proof}
The proof is the same as for the corresponding Lemma in \citet{Cappe}, but with slight modifications to allow for conditioning on controlsl
\begin{enumerate} 
  \item Letting $A=X$ in the strong mixing conditions we find that for all $u$
  \[
    \varsigma^-(y) \leq \int Q^u(x,dx')g(x',y) \leq \varsigma^+(y).
  \]
  We also have $\beta_{k|n}(x)$ expressed as
  \begin{align*}
    & = \int_{x_{k+1}} \cdots \int_{x_n} Q^{u_k}(x,dx_{k+1})g(x_{k+1},y_{k+1}) \prod_{l=k+2}^n Q^{u_{l-1}}(x_{l-1},dx_l)g(x_l,y_l) \\
   &  \hspace{1cm}  =\int_{x_{k+1}} Q^{u_k}(x,dx_{k+1})g(x_{k+1},y_{k+1}) \\
& \hspace{1.2cm}   \times\int_{x_{k+2}} \cdots \int_{x_n} Q^{u_{k+1}}(x_{k+1},dx_{k+2}) g(x_{k+2},y_{k+2}) \prod_{l=k+3}^n Q^{u_{l-1}}(x_{l-1},dx_l)g(x_l,y_l) \\
&  \hspace{1cm}   \leq \varsigma^+(y_{k+1}) \sup_{x_{k+1}} \int_{x_{k+2}} \cdots \int_{x_n} Q^{u_{k+1}}(x_{k+1},dx_{k+2}) g(x_{k+2},y_{k+2}) \prod_{l=k+3}^n Q^{u_{l-1}}(x_{l-1},dx_l)g(x_l,y_l) \\
 &  \hspace{1cm} = \varsigma^+(y_{k+1}) \sup_{x} \beta_{k+1|n}(x) \leq \prod_{j=k+1}^n \varsigma^+(y_j).
  \end{align*}
  The other inequality is similar.

  \item Using the recursion for the backward variables we find
  \begin{align*}
  &  \int_x \nu(dx)\beta_{k|n}(y_{(k+1):n},u_{k:n}) \\
 & \hspace{1cm}  = \int_x \int_{x_{k+1}} \nu(dx)Q^{u_k}(x,x_{k+1})g(x_{k+1},y_{k+1})\beta_{k+1|n}(y_{(k+2):n},u_{(k+1):n},dx_{k+1}) \\
  &  \hspace{1cm} = \int_{x_{k+1}} \left[\int_x  \nu(dx)Q^{u_k}(x,x_{k+1})g(x_{k+1},y_{k+1})\right]\beta_{k+1|n}(y_{(k+2):n},u_{(k+1):n},dx_{k+1}) \\
  &  \hspace{1cm} \leq \int_{x_{k+1}} \left[\int_x  \nu(dx)\varsigma^+(y_{k+1})K^{u_k}(y_{k+1},x_{k+1})\right]\beta_{k+1|n}(y_{(k+2):n},u_{(k+1):n},dx_{k+1}) \\
 &  \hspace{1cm} =\varsigma^+(y_{k+1})\int_{x_{k+1}} K^{u_k}(y_{k+1},x_{k+1})\beta_{k+1|n}(y_{(k+2):n},u_{(k+1):n},dx_{k+1}).
  \end{align*}
  We get a similar inequality for $\varsigma^-$. Also note that the last integral doesn't depend on $\nu$, so it cancels when we take the ratio. The result follows.

  \item We have that
  \begin{align*}
    F_{k|n}[y_{(k+1):n},u_{k:n}](x,A) & = \frac{\int_A Q^{u_k}(x,dx_{k+1})g(x_{k+1},y_{k+1}) \beta_{k+1|n}(x_{k+1})}{\int Q^{u_k}(x,dx_{k+1})g(x_{k+1},y_{k+1}) \beta_{k+1|n}(x_{k+1})} \\
   &  \leq\frac{\varsigma^+(y_{k+1})}{\varsigma^-(y_{k+1})}\cdot\frac{\int_A K^{u_k}(y_{k+1},dx_{k+1}) \beta_{k+1|n}(x_{k+1})}{\int K^{u_k}(y_{k+1},dx_{k+1})  \beta_{k+1|n}(x_{k+1})}
  \end{align*}
  and we can set
  \[
    \lambda_{k|n}(y_{(k+1):n},u_{k:n},A)= \frac{\int_A K^{u_k}(y_{k+1},dx_{k+1}) \beta_{k+1|n}(x_{k+1})}{\int K^{u_k}(y_{k+1},dx_{k+1})  \beta_{k+1|n}(x_{k+1})}.
  \]

  \item Using $(iii)$ we find that
  \[
   F_{k|n}[y_{(k+1):n},u_{k:n}](x,A) \geq \frac{\varsigma^-(y_{k+1})}{\varsigma^+(y_{k+1})} \lambda_{k|n}(y_{(k+1):n},u_{k:n},A)
  \]
  and thus Assumption $4.3.12$ holds and Lemma $4.3.13$ gives
  \[
   \delta(F_{k|n}) \leq \rho_0(y_{k+1}) =  1-\frac{\varsigma^-(y_{k+1})}{\varsigma^+(y_{k+1})}.
  \]
\end{enumerate}
\end{proof}

\begin{theorem}[Proposition 4.3.23 in \citet{Cappe}]
Under the strong mixing conditions the following holds
\begin{enumerate} 
  \item We let $\nu$ and $\nu'$ be two different initial distributions for $X_0$. Now for $k\leq n$
  \begin{align*}
   & ||\phi_{\nu,k|n}[y_{0:n},u_{0:(n-1)}] - \phi_{\nu',k|n}[y_{0:n},u_{0:(n-1)}]||_{TV} \\
   & \hspace{1cm} \leq \left[\prod_{j=1}^k\rho_0(y_j)\right]||\phi_{\nu,0|n}[y_{0:n},u_{0:(n-1)}] - \phi_{\nu',0|n}[y_{0:n},u_{0:(n-1)}]||_{TV} \\
   & \hspace{1cm} \leq 2\left[\prod_{j=1}^k\rho_0(y_j)\right].
  \end{align*}

  \item
  For any non-negative integers $j,k,n$ such that $ j \leq k \leq n$
  \[
   ||P_{\nu}(X_k \in \cdot \,|y_{0:n},u_{0:(n-1)}) - P_{\nu}(X_k \in \cdot \,| Y_{j:n}, u_{j:(n-1)})||_{TV}
   \leq 2 \prod_{i=j}^k \rho_0(y_i)
  \]
  where $\nu$ is the initial distribution of $X_0$.
\end{enumerate}
\end{theorem}

\begin{proof}
 \begin{enumerate} 
  \item Earlier we had
  \[
 ||\phi_{\nu,k|n} - \phi_{\nu',k|n}||_{TV} \leq \prod_{i=1}^k\delta\left( F_{i-1|n}(y_{i:n},\cdot)\right)||\phi_{\nu,0|n} - \phi_{\nu',0|n}||_{TV}
\]
and the first inequality now follows from the Lemma $4.3.22$ part $(iv)$. The factor ``2" follows from using the triangle inequality on the difference of two probability measures.
\item This is just like part $(i)$ except we consider different initial distributions for $X_j$.
 \end{enumerate}
\end{proof}

\subsection{Bounds on score function, Chapter 12 in \citet{Cappe}}


Set $h_{k,x}(\theta) = \log \left[\int g(x_k,Y_k)P(X_k \in
dx_k|Y_{0:(k-1)},u_{0:(k-1)},X_0=x)\right]$. Then our usual loglikelihood is $l_{x,n}(\theta)
= \sum_{k=0}^n h_{k,x}(\theta)$

We now wish to use the expression for $\frac{\partial}{\partial \theta}
l(\theta)$ derived in the last section.
We have that $\frac{\partial}{\partial \theta} l_{x,n}(\theta) = \sum_{k=0}^n \dot{h}_{k,x}(\theta)$ but also
\[
 \frac{\partial}{\partial \theta} l_{x,n}(\theta) = \frac{\partial}{\partial
\theta} l_{x,0}(\theta) +\sum_{k=1}^n\left\{\frac{\partial}{\partial \theta}
l_{x,k}(\theta)-\frac{\partial}{\partial \theta} l_{x,k-1}(\theta)\right\}
\]
This gives an alternative expression of $\dot{h}_{k,x}$. We get $\dot{h}_{0,x}(\theta)= \frac{\partial}{\partial \theta} \log g(x_0,Y_0)$ and for $k\geq 1$
\begin{align*}
 \dot{h}_{k,x}(\theta) & =\frac{\partial}{\partial \theta}
l_{x,k}(\theta)-\frac{\partial}{\partial \theta} l_{x,k-1}(\theta) \\
&  = E\left[\sum_{i=1}^k
\phi(X_{i-1},X_i,Y_i)|Y_{1:k},u_{0:(k-1)},X_0=x\right] \\
& \hspace{0.5cm} -E\left[\sum_{i=1}^{k-1}
\phi(X_{i-1},X_i,Y_i)|Y_{1:(k-1)},u_{0:(k-2)},X_0=x\right]
\end{align*}
This expression can be generalized to starting the process at other values than zero;
\begin{align*}
 \dot{h}_{k,m,x}(\theta) & = \log \left[\int g(x_k,Y_k)P(X_k \in
dx_k|Y_{m:(k-1)},u_{m:(k-1)},X_m=x)\right] \\
&  = E\left[\sum_{i=m+1}^k
\phi(X_{i-1},X_i,Y_i)|Y_{(m+1):k},u_{m:(k-1)},X_{m}=x\right] \\
& \hspace{0.5cm} - \hspace{0.5cm} E\left[\sum_{i=m+1}^{k-1}
\phi(X_{i-1},X_i,Y_i)|Y_{(m+1):(k-1)},u_{m:(k-2)},X_{m}=x\right]
\end{align*}
This is done in \citet{Cappe} to extend the process to minus infinity ($m \rightarrow -\infty$). We don't extend the process to infinity, but rather think of $m$ as indicating lack of information, that is assuming that the process starts at $X_m$.

We now prove a modified Lemma $12.5.3$ where we use the expression developed above.

\begin{theorem}[Lemma 12.5.3 in \citet{Cappe} modified]
Assuming strong mixing conditions. Then for $k\geq 1$ \citet{Cappe} prove the following inquality in the HMM case:
 \[
  (E |\dot{h}_{k,-m,x}(\theta) - \dot{h}_{k,\infty}(\theta)|^2)^{1/2}
  \leq 12\left(E\left[\sup_{x,x' \in X}
  |\phi_{\theta}(x,x',Y_1)|^2\right]\right)^{1/2}
  \frac{\rho^{(k+m)/2-1}}{1-\rho}.
 \]
 We don't extend the process to $-\infty$, but rather starting at $X_0$ and we prove the following inequality, also for $k\geq 1$
  \[
  (E |\dot{h}_{k,0,x_0}(\theta) - \dot{h}_{k,m,x}(\theta)|^2)^{1/2}
  \leq 8\sup_{x,x'\in X,u \in U, y \in Y} ||\phi_{\theta}(x,x',y,u)||\frac{\rho^{(k-m)/2-1}}{1-\rho}
 \]
 where $\rho = \max_{y \in Y} \rho_0(y)$ (See Theorem \ref{lemma 4.3.22}).
\label{lemma 12.5.3m}
\end{theorem}

\begin{proof}
 From the representation derived above for $\dot{h}$ we have
  \begin{align}
    \dot{h}_{k,0,x_0}(\theta) & = E\left[\sum_{i=1}^k
    \phi(X_{i-1},X_i,Y_i,u_{i-1})|Y_{1:k},u_{0:(k-1)},X_0=x_0\right] \label{tag1} \\
    & \hspace{0.5cm} -E\left[\sum_{i=1}^{k-1}
    \phi(X_{i-1},X_i,Y_i,u_{i-1})|Y_{1:(k-1)},u_{0:(k-2)},X_0=x_0\right] \label{tag2}
  \end{align}
 and
  \begin{align}
    \dot{h}_{k,m,x}(\theta) & = E\left[\sum_{i=m+1}^k
    \phi(X_{i-1},X_i,Y_i,u_{i-1})|Y_{(m+1):k},u_{m:(k-1)},X_m=x\right] \label{tag3} \\
    & \hspace{0.5cm} -E\left[\sum_{i=m+1}^{k-1}
    \phi(X_{i-1},X_i,Y_i,u_{i-1})|Y_{(m+1):(k-1)},u_{m:(k-2)},X_m=x\right] \label{tag4}
  \end{align}
Just like in the proof of Lemma $12.5.3$ in \citet{Cappe} we match together different pairs of terms within the sums, depending on their index $i$. More specifically for $i=k$ we match together the terms where $i=k$ in \eqref{tag1} and \eqref{tag3}. For $\frac{k+m}{2} \leq i <k$ we match the terms in \eqref{tag1} with \eqref{tag3} and the terms in \eqref{tag2} with those in \eqref{tag4}. For $m+1\leq i <\frac{k+m}{2}$ we match terms in \eqref{tag1} with terms in \eqref{tag2} and terms in \eqref{tag3} with those in \eqref{tag4}.  That leaves $i\in {1,\ldots,m}$ in $\dot{h}_{k,0,x_0}$ where we match \eqref{tag1} and \eqref{tag2}.

  If we look at the case where \eqref{tag1} is matched with \eqref{tag3} we have
  \begin{align*}
  & ||E[\phi_{\theta}(X_{i-1},X_i,Y_i,u_{i-1})|Y_{(m+1):k},u_{m:(k-1)},X_{m}=x]
   -E[\phi_{\theta}(X_{i-1},X_i,Y_i,u_{i-1})|Y_{1:k},u_{0:(k-1)}]|| \\
  & = |\int_{x_{m}}\int_{x_{i-1}}\int_{x_i} \phi_{\theta}(x_{i-1},x_i,Y_i,u_i) F_{i-1}(x_{i-1},dx_i)
   P_{\theta}(X_{i-1}\in dx_{i-1}|Y_{(m+1):k},u_{m:(k-1)},X_{m}=x) \\
  & \hspace{1cm} \times[\delta_x(dx_{m}) - P_{\theta}(X_{m}\in dx_{m}|Y_{1:k},u_{0:(k-1)})]| \\
  & \leq 2 \sup_{x,x'\in X, u \in U} ||\phi_{\theta}(x,x',Y_i,u)||\rho^{(i-1)-m}
  \end{align*}
  where $F_{i-1} = F_{i-1;\theta}[y_{i:k},u_{(i-1):k}]$ is the Forward Smoothing Kernel, and the inequality stems from Proposition $4.3.23$ $(i)$ where the second line can be thought of as two different initial distributions for $X_m$, and the kernel $F$ is bounded by $1$.

  Matching \eqref{tag2} with \eqref{tag4} is similar.
  For matching \eqref{tag1} with \eqref{tag2} and \eqref{tag3} with \eqref{tag4} we need a ``Backwards bound";
  \[
   ||P_{\theta}(X_i\in  \cdot \,|Y_{(m+1):k},u_{m:(k-1)},X_{m}=x) - P_{\theta}(X_i \in \cdot \,|Y_{(m+1):(k-1)},u_{m:(k-2)},X_{m}=x)||_{TV}
   \]
   \[
    \leq 2\rho^{k-1-i}
  \]
  that is established below, see Theorem \ref{prop 12.5.4}. For matching \eqref{tag3} with \eqref{tag4} we get
  \begin{align*}
   & ||E_{\theta}[\phi_{\theta}(X_{i-1},x_i,Y_i,u_{i-1})|Y_{(m+1):k},u_{m:(k-1)},X_{m} = x]
   -E_{\theta}[\phi_{\theta}(X_{i-1},x_i,Y_i,u_{i-1})|Y_{(m+1):(k-1)},u_{m:(k-2)},X_{m} = x]|| \\
  & = |\int_{x_{i-1}}\int_{x_i} \phi_{\theta}(x_{i-1},x_i,Y_i,u_{i-1}) B_{i}(x_i,dx_{i-1}) \\
  & \hspace{1cm} \times[P_{\theta}(X_{i}\in dx_{i}|Y_{(m+1):k},u_{m:(k-1)},X_{m}=x) - P_{\theta}(X_{i}\in dx_{i}|Y_{(m+1):(k-1)},u_{m:(k-2)},X_{m}=x)]| \\
  & \leq 2 \sup_{x,x'\in X,u \in U} ||\phi_{\theta}(x,x',Y_i,u)||\rho^{(k-1)-i}
  \end{align*}
  where $B_i$ is the Backwards Smoothing Kernel described below.
  Matching \eqref{tag1} with \eqref{tag2} is a special case of the above.

  Going back to our original objective, we have
  \[
   \left(E_{\theta}||\dot{h}_{k,m,x}(\theta) - \dot{h}_{k,0,x_0}(\theta)||^2\right)^{1/2}
   =\left(E||\sum a_i||^2\right)^{1/2}
  \]
  where $\sum a_i$ is a sum over the pairs we considered above. Now by Minkowski's inequality we have
  \[
   \left(E||\sum a_i||^2\right)^{1/2} \leq \sum \left( E||a_i||^2\right)^{1/2}
  \]
  Now we have that $||a_i|| \leq 2 \sup_{x,x'\in X,u \in U} ||\phi_{\theta}(x,x',Y_i,u)||\rho^{b_i}$ where $b_i$ is the power of $\rho$ associated with $a_i$ and therefore
  \[
    \left(E_{\theta}||\dot{h}_{k,m,x}(\theta) - \dot{h}_{k,0,x_0}(\theta)||^2\right)^{1/2} \leq \sum 2\left(E \sup_{x,x'\in X,u \in U} ||\phi_{\theta}(x,x',Y_i,u)||^2\right)^{1/2}\rho^{b_i}
  \]
  At this point \citet{Cappe} argue that since in their case the process was started at inifinity and the process is homogeneous the expected value over $Y_i$ is always the same by stationarity, and $Y_i$ can be exchanged by $Y_1$. Since arguing for stationarity is difficult in a POMPDP setting, we also take the supremum over $Y$ and remember that this set is also finite so that
  \begin{align*}
   \left(E_{\theta}||\dot{h}_{k,m,x}(\theta) - \dot{h}_{k,0,x_0}(\theta)||^2\right)^{1/2} & \leq 2\left( \sup_{x,x'\in X,u \in U, y \in Y} ||\phi_{\theta}(x,x',y,u)||^2\right)^{1/2}\sum \rho^{b_i} \\
   & = 2 \sup_{x,x'\in X,u \in U, y \in Y} ||\phi_{\theta}(x,x',y,u)||\sum \rho^{b_i}
  \end{align*}
We now deal with the sum of $\rho$ to different powers.

From $i=k$ we have $\rho^{k-1-m}$ where we matched \eqref{tag1} with \eqref{tag3}. For $\frac{k+m}{2}\leq i < k$ we have  $2\rho^{i-1-m}$ where we matched \eqref{tag1} with \eqref{tag3} and \eqref{tag2} with \eqref{tag4}. For $m+1 \leq i <\frac{k+m}{2}$ we have $2\rho^{k-1-i}$ from matching \eqref{tag1} with \eqref{tag2} and \eqref{tag3} with \eqref{tag4}. Finally for $1\leq i \leq m$ we have $\rho^{k-1-i}$ from matching \eqref{tag1} with \eqref{tag2}. This gives
\[
 \sum \rho^{b_i} = \rho^{k-1-m}  +  \sum_{i=(k+m)/2}^{k-1}2\rho^{i-1-m} + \sum_{i=m+1}^{(k+m)/2-1} 2\rho^{k-1-i} +\sum_{i=1}^m\rho^{k-1-i}
\]
\[
 \leq 2\sum_{i = (k+m)/2}^{\infty} \rho^{i-1-m} + 2\sum_{i=-\infty}^{(k+m)/2-1}\rho^{k-1-i}
\]
\[
 =2 \frac{\rho^{(k-m)/2-1}}{1-\rho} + 2\frac{\rho^{(k-m)/2}}{1-\rho} \leq 4\frac{\rho^{(k-m)/2-1}}{1-\rho}
\]
Thus, finally we have
  \[
   \left(E_{\theta}||\dot{h}_{k,m,x}(\theta) - \dot{h}_{k,0,x_0}(\theta)||^2\right)^{1/2}
   \leq 8\sup_{x,x'\in X,u \in U, y \in Y} ||\phi_{\theta}(x,x',y,u)||\frac{\rho^{(k-m)/2-1}}{1-\rho}
  \]

\end{proof}

\begin{theorem}[Proposition 12.5.4 modified]
\label{prop 12.5.4}
  \[
   ||P_{\theta}(X_i\in  \cdot \,|Y_{(m+1):k},u_{m:(k-1)},X_{m}=x)
   - P_{\theta}(X_i \in \cdot \,|Y_{(m+1):(k-1)},u_{m:(k-2)},X_{m}=x)||_{TV} \leq 2\rho^{k-1-i}
  \]
\end{theorem}
\begin{proof}
 The idea behind this proof is to replicate all the results derived so far for the Backward Smoothing Kernel. That is, conditional on $Y_{(m+1):k}$, $u_{m:(k-1)}$ and $X_m = x_m$ the time-reversed process $X$ is a non-homogeneous Markov Chain, where the conditional probability of moving from $X_{j+1}$ to $X_j$ given all the observations $Y_{(m+1):(k-1)}$, controls $u_{m:(k-2)}$ and initial condition ends up only depending on $Y_{(m+1):j}$, $u_{m:j}$ and the initial condition, and is governed by the Backwards Smoothing Kernel given by
 \[
  B_{x_m,j}[u_{(m+1):j},u_{m:j}](x,f)
  =\frac{\int\cdots\int \prod_{r=m+1}^j Q^{u_{r-1}}(x_{r-1},dx_r)g(x_r,y_r)f(x_j)Q^{u_j}(x_j,x)}{\int\cdots\int \prod_{r=m+1}^j Q^{u_{r-1}}(x_{r-1},dx_r)g(x_r,y_r)Q^{u_j}(x_j,x)}
 \]
Just as we did in Lemma $4.3.22$ we can show
\begin{align*}
 \frac{\varsigma^-(y_j)}{\varsigma^+(y_j)}\nu_{x_m, j}[y_{m+1},u_{m:j}]
 & \leq B_{x_m,j}[y_{(m+1):j},u_{m:j}](x_j,\cdot \,) \\
 & \frac{\varsigma^+(y_j)}{\varsigma^-(y_j)}\nu_{x_m, j}[y_{m+1},u_{m:j}]
\end{align*}
where
\[
 \nu_{x_m, j}[y_{m+1},u_{m:j}](f)
 =\frac{\int\cdots\int \prod_{r=m+1}^j Q^{u_{r-1}}(x_{r-1},dx_r)g(x_r,y_r)f(x_j)}{\int\cdots\int \prod_{r=m+1}^j Q^{u_{r-1}}(x_{r-1},dx_r)g(x_r,y_r)}.
\]
As we showed there this gives
\[
 \delta(B_{x_m,j}) \leq 1-\frac{\varsigma^-(y_j)}{\varsigma^+(y_j)}
\]
We now get that the 2 smoothers we are interested in can be thought of as smoothers of the reversed Markov Chain from $k-1$ to $m$ with 2 different initial distributions for $X_{k-1}$, the starting position. We get
\begin{align*}
&  ||P_{\theta}(X_i\in  \cdot \,|Y_{(m+1):k},u_{m:(k-1)},X_{m}=x)
   - P_{\theta}(X_i \in \cdot \,|Y_{(m+1):(k-1)},u_{m:(k-2)},X_{m}=x)||_{TV} \\
& \leq ||P_{\theta}(X_{k-1}\in  \cdot \,|Y_{(m+1):k},u_{m:(k-1)},X_{m}=x)
   - P_{\theta}(X_{k-1} \in \cdot \,|Y_{(m+1):(k-1)},u_{m:(k-2)},X_{m}=x)||_{TV} \\
& \hspace{0.5cm} \times \prod_{j=i+1}^{k-1} \delta(B_{x_m,j}) \\
& \leq 2\prod_{j=i+1}^{k-1} \rho_0(y_j) \leq 2 \rho^{k-1-i}
\end{align*}
(where $\rho = \max_{y\in Y} \rho_0(y)$).
\end{proof}

\end{document}